\documentclass[a4paper,10pt]{article}
\usepackage{amsthm,amssymb, amsmath}
\usepackage{graphicx}
\usepackage{graphics}
\usepackage{times}
\usepackage{amsmath}
\usepackage{amsfonts}
\usepackage{color}

\newtheorem{Prop}{Proposition}
\newtheorem{Def}{Definition}
\newtheorem{Thm}{Theorem}
\newtheorem{Lem}{Lemma}
\newtheorem{Cor}{Corollary}

\newcommand{\mb}{\mathbf}
\newcommand{\lsb}{\left (}
\newcommand{\rsb}{\right )}
\newcommand{\lmb}{\left \{}
\newcommand{\rmb}{\right \}}

\newcommand\independent{\protect\mathpalette{\protect\independenT}{\perp}}
\def\independenT#1#2{\mathrel{\rlap{$#1#2$}\mkern2mu{#1#2}}}

\begin{document}

\title{~\textbf{Optimality of Binning for Distributed Hypothesis Testing}~}

\date{}

\author{Md. Saifur Rahman and Aaron B. Wagner$^{\footnote{Both authors are with the School of Electrical and Computer Engineering, Cornell University, Ithaca, NY, 14853, USA. (Email: mr534@cornell.edu, wagner@ece.cornell.edu.)}}$}

\maketitle

\begin{abstract}
We study a hypothesis testing problem in which data is compressed distributively and sent to a detector that seeks to decide between two possible distributions for the data. The aim is to characterize all achievable encoding rates and exponents of the type 2 error probability when the type 1 error probability is at most a fixed value. For related problems in distributed source coding, schemes based on random binning perform well and often optimal. For distributed hypothesis testing, however, the use of binning is hindered by the fact that the overall error probability may be dominated by errors in binning process. We show that despite this complication, binning is optimal for a class of problems in which the goal is to ``test against conditional independence.'' We then use this optimality result to give an outer bound for a more general class of instances of the problem.
\end{abstract}


\section{Introduction}
Consider the problem of measuring the traffic on two links
in a communication network and inferring whether
the two links are carrying any common traffic~\cite{He, Ameya}. Evidently,
this inference cannot be made by inspecting the measurements from one
of the links alone, except in the extreme situation in which that
link carries no traffic at all. Thus it is necessary to transport
the measurements from one of the links to the other, or to
transport both measurements to a third location. The measured data is
potentially high-rate, however, so this transportion may require that
the data be compressed. This raises the question of how to compress
data when the goal is not to reproduce it \emph{per se}, but rather
to perform inference. A similar problem arises when inferring
the speed of a moving vehicle from the times that it passes
certain waypoints.

These problems can be modeled mathematically by the setup
depicted in Fig. \ref{fig:Fig1}, which we call the $L$-encoder general hypothesis testing problem. A vector source $(X_1,\dots,X_L, Y)$ has different joint distributions $P_{X_1,\dots,X_L,Y}$ and $Q_{X_1,\dots,X_L,Y}$ under two hypotheses $H_0$ and $H_1$, respectively. Encoder $l$ observes an i.i.d.\ string distributed according to $X_l$ and sends a message to the detector at a finite rate of $R_l$ bits per observation using a noiseless channel. The detector, which has access to an i.i.d.\ string distributed according to $Y$, makes
a decision between the hypotheses. The detector may make two types of error: the type 1 error ($H_0$ is true but the detector decides otherwise) and the type 2 error ($H_1$ is true but the detector decides otherwise). The type 1 error probability is upper bounded by a fixed value. The type 2 error probability decreases exponentially fast, say with an exponent $E$, as the length of the i.i.d.\ strings increases. The goal is to characterize the rate-exponent region of the problem, which is the set of all achievable rate-exponent vectors $(R_1,\dots,R_L,E)$, in the regime in which the type 1 error probability is small. This problem was first introduced by Berger~\cite{Berger} (see also \cite{Han}) and arises naturally in many applications. Yet despite these applications, the theoretical understanding of this problem is far from complete, especially when compared with its sibling, distributed source coding, where random binning has been shown to be a key ingredient in many optimal schemes.
\begin{figure}[htp]
\centering
\includegraphics[width=3.5in, totalheight=0.2\textheight,viewport=10 60 730 500,clip]{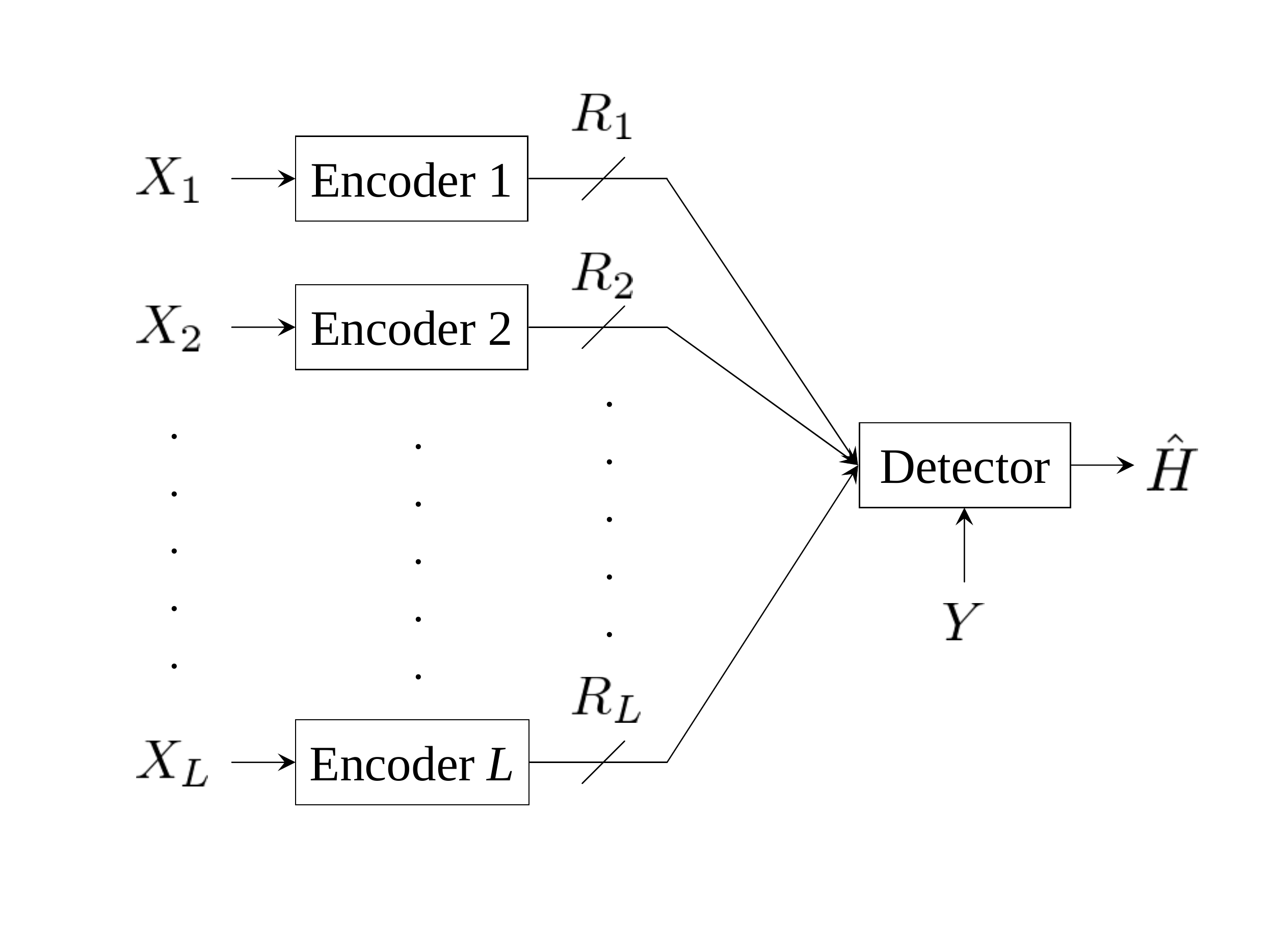}
\caption{$L$-encoder general hypothesis testing}\label{fig:Fig1}
\end{figure}

Note that if one of the variables in the set $(X_1,\ldots,X_L,Y)$ has a different marginal distribution under $P_{X_1,\dots,X_L,Y}$ and $Q_{X_1,\dots,X_L,Y}$, then one of the terminals can detect the underlying hypothesis with an exponentially-decaying type 2 error probability, even without receiving any information from the other terminals, and could communicate this decision to other terminals by broadcasting a single bit. Motivated by the applications mentioned above, we shall focus our attention on the case in which the variables $X_1,\ldots,X_L,Y$ have the same marginal distibutions under both hypotheses.

Ahlswede and Csisz\'{a}r \cite{Ahl} studied a special case of this problem in which $L=1$. They presented a scheme in which the encoder sends a quantized value of $X_1$ to the detector which uses it to perform the test with the help of $Y$. They showed that their scheme is optimal for a ``test against independence.'' Their scheme was later improved by Han~\cite{Han1} and Shimokawa-Han-Amari~\cite{Han2}. In the latter improvement, the encoder first quantizes $X_1$, then bins the quantized value using a Slepian and Wolf encoder \cite{Slepian}. The detector first decodes the quantized value with the help of $Y$ and then performs a likelihood ratio test. In this scheme, type 2 errors can occur in two different ways: the binning can fail so that the receiver decodes the wrong codeword and therefore makes an incorrect decision, or the true codeword can be decoded correctly yet be atypically distributed with $Y$, again resulting in an incorrect decision. Moreover, there is a tension between these two forms of error. If the codeword is a high fidelity representation of $X_1$, then binning errors are likely, yet the detector is relatively unlikely to make an incorrect decision if it decodes the codeword correctly. If the codeword is a low fidelity representation, then binning errors are unlikely, but the detector is more likely to make an incorrect decision when it decodes correctly.

\begin{figure}[htp]
\centering
  \includegraphics[width=3.8in]{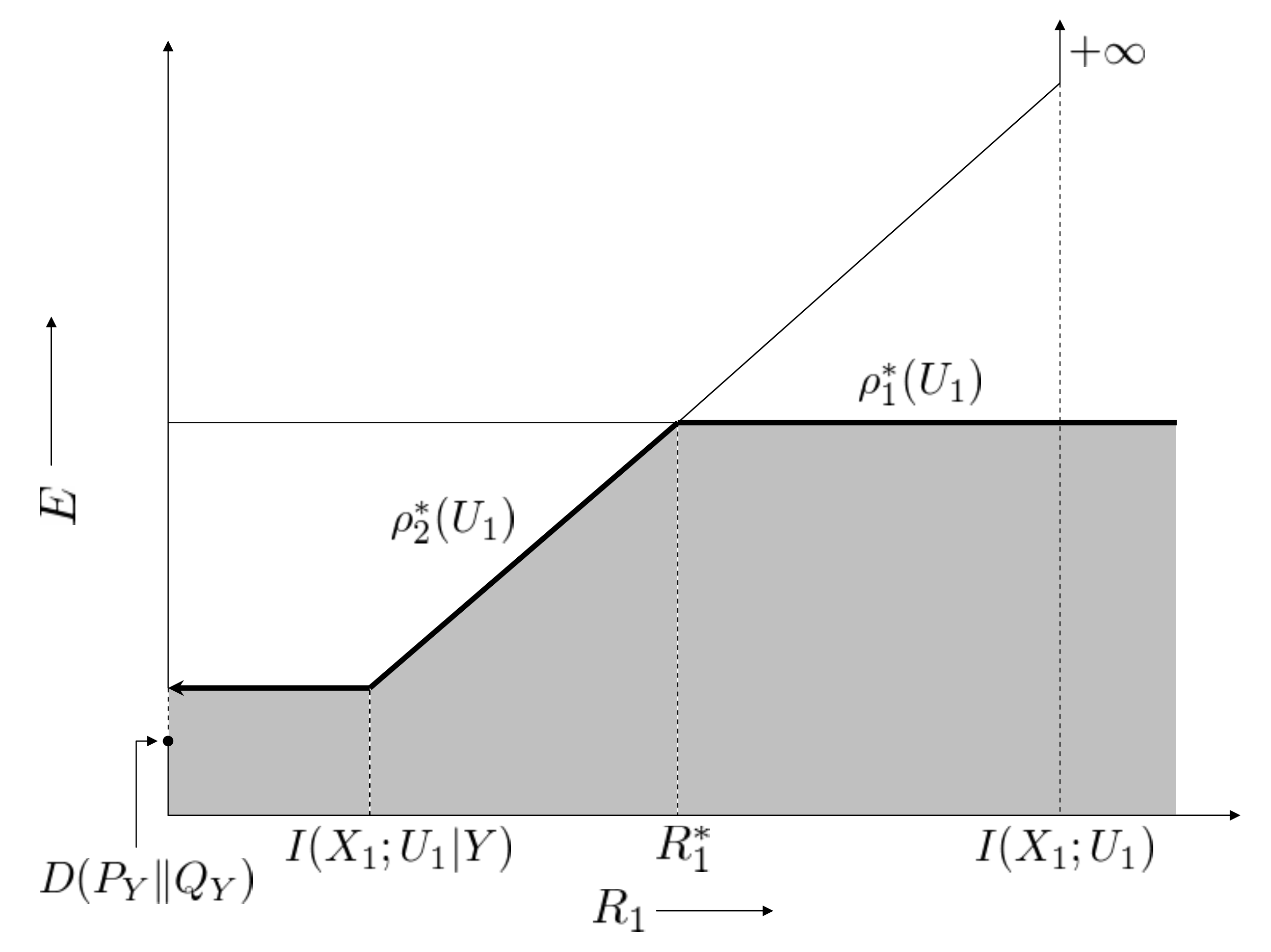}
\caption{Shimokawa-Han-Amari achievable region for a fixed channel $P_{U_1|X_1}$} \label{fig:Fig2}
\end{figure}
Fig. \ref{fig:Fig2} illustrates this tradeoff for a fixed test channel $P_{U_1|X_1}$ used for quantization. All mutual information quantities are computed with respect to $P$. $\rho_2^{*}(U_1)$ and $\rho_1^{*}(U_1)$ are the exponents associated with type 2 errors due to binning errors and assuming correct decoding of the codeword, respectively. Formulas for each are available in ~\cite{Han}. For low rates, binning errors are common and $\rho_2^{*}(U_1)$ dominates the overall exponent. For high rates, binning errors are uncommon and $\rho_1^{*}(U_1)$ dominates the overall exponent. To achieve the overall performance, the test channel should be chosen so that these two exponents are equal; if they are not, then making the test channel slightly more or less noisy will yield better performance. A similar tradeoff arises in the analysis of error exponents of binning-based schemes for the Wyner-Ziv problem ~\cite{Ben, Ben2, Ben3,Kochman} and in the design of short block-length codes for Wyner-Ziv or joint source-channel coding. Evidently the benefit accrued from binning is reduced when one considers error exponents, as opposed to when the design criterion is vanishing error probability or average distortion, because the error exponent associated with the binning process itself may dominate the overall performance.

The Shimokawa-Han-Amari scheme uses random, unstructured binning. It is known from the lossless source coding literature that structured binning schemes can strictly improve upon unstructured binning schemes in terms of the error exponents~\cite{Csiszar, Ben1, Csiszar1}. Thus, two questions naturally arise:
\begin{enumerate}
\item Is the tradeoff depicted in Fig. \ref{fig:Fig2} fundamental to the problem or an artifact of a suboptimal scheme?
\item Can the scheme be improved by using structured binning?
\end{enumerate}

We conclusively answer both questions and show that unstructured binning is optimal in several important cases. We begin by considering a special case of the problem that we call $L$-encoder hypothesis testing against conditional independence. Here $Y$ is replaced by a three-source $(X_{L+1},Y,Z)$ such that $Z$ induces conditional independence between $(X_1,\dots,X_L,X_{L+1})$ and $Y$ under $H_1$. In addition, $(X_1,\dots,X_L,X_{L+1},Z)$ and $(Y,Z)$ have the same distributions under both hypotheses. This problem is a generalization of the single-encoder test against independence studied by Ahlswede and Csisz\'{a}r \cite{Ahl},

For this problem we provide an achievable region, based on a scheme we call Quantize-Bin-Test, that reduces to the Shimokawa-Han-Amari region for $L = 1$ yet is significantly simpler. We also introduce an outer bound similar to the outer bound for the distributed rate-distortion problem given by Wagner and Anantharam \cite{Wagner3}. The idea is to introduce an auxiliary random variable that induces conditional independence between the sources. This technique of obtaining an outer bound has been used to prove results in many distributed source coding problems \cite{Wagner3, Wagner1, Wagner, Wagner2, Ozarow, Wang}.

The inner (achievable) and outer bounds are shown to match in three examples. The first is the case in which there is only one encoder ($L = 1$). Although this problem is simply the conditional version of the test against independence studied by Ahlswede and Csisz\'{a}r~\cite{Ahl}, the conditional version is much more complicated due to the necessary introduction of binning. It follows that the Shimokawa-Han-Amari scheme is optimal for $L = 1$, providing what appears to be the first nontrivial optimality result for this scheme. This problem arises in detecting network flows in the presence of common cross-traffic that is known to the detector. Here $X_1$ represents the network traffic measured at a remote location, $Y$ is the traffic measured at the detector, and $Z$ represents the cross-traffic. The goal is to detect the presence of common traffic beyond $Z$, i.e., to determine whether $Z$ captures all of the dependence between $X_1$ and $Y$.

The second is a problem inspired by a result of Gel`fand and Pinsker~\cite{Gelfand}. We refer to this as the Gel`fand and Pinsker hypothesis testing against independence problem, the setup of which is shown in Fig. \ref{fig:Fig3}. Here $X_{L+1}$ and $Z$ are deterministic and there is a source $X$ which under $H_0$ is the minimum sufficient statistic for $Y$ given $(X_1,\dots,X_L)$ such that $X_1,\dots,X_L, Y$ are conditionally independent given $X$. We characterize the set of rate vectors $(R_1,\dots,R_L)$ that achieve the centralized exponent $I(X;Y)$. We show that the Quantize-Bin-Test scheme is optimal for this problem.

\begin{figure}[htp]
\centering
  \includegraphics[width=3.5in]{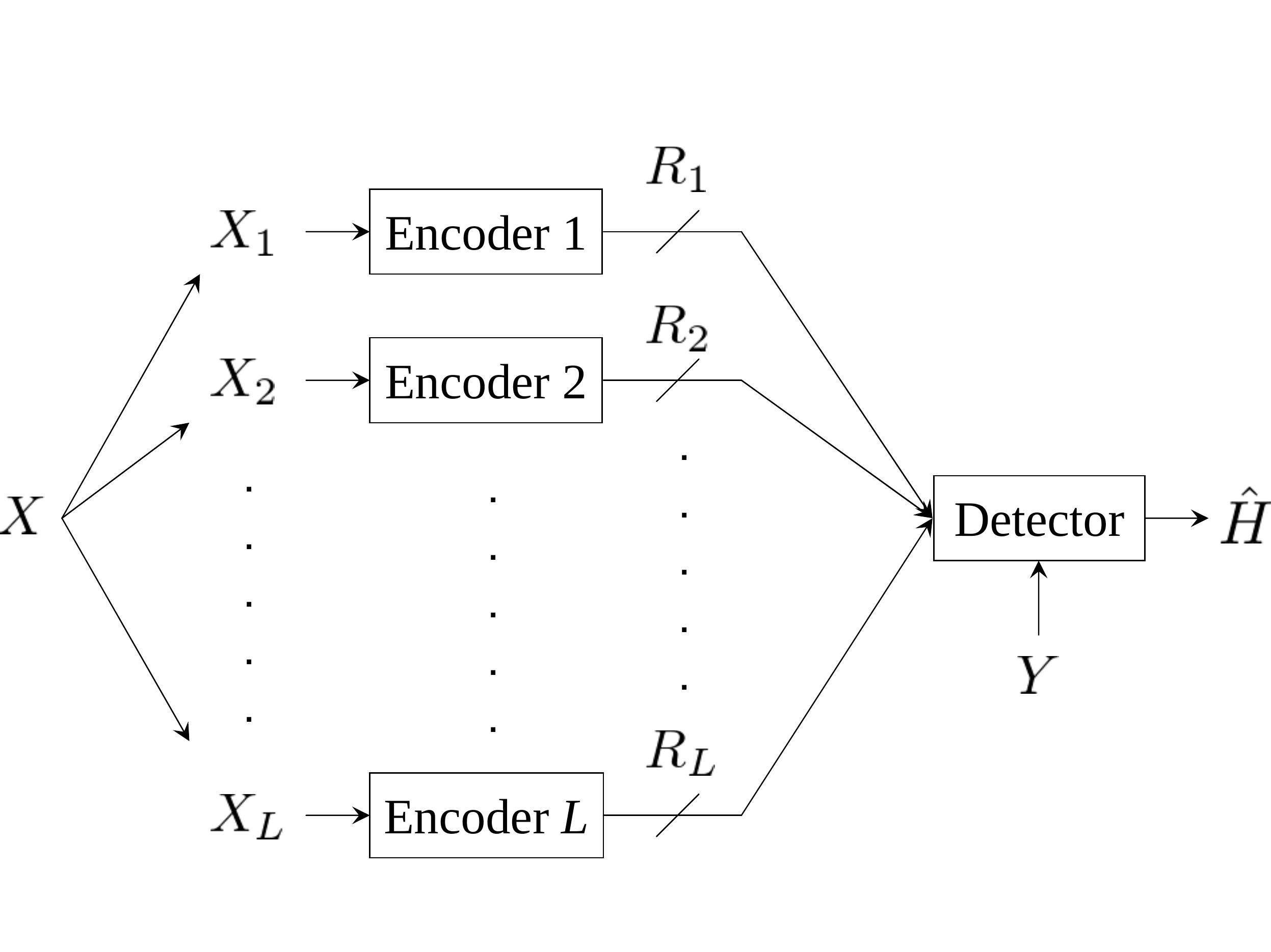}
\caption{Gel`fand and Pinsker hypothesis testing against independence}\label{fig:Fig3}
\end{figure}
The third is the Gaussian many-help-one hypothesis testing against independence problem, the setup of which is shown in Fig. \ref{fig:Fig4}. Here the sources are jointly Gaussian and there is another scalar Gaussian source ${X}$ observed by the main encoder which sends a message to the detector at a rate $R$. The encoder observing $X_l$ is now referred to as the helper $l$. We characterize the rate-exponent region of this problem in a special case when $X_1,\dots,X_L,Y$ are conditionally independent given $X$. We use results on related source coding problem by Oohama \cite{Oohama2005} and Prabhakaran \emph{et al.} \cite{Vinod} to obtain an outer bound, which we show is achieved by the Quantize-Bin-Test scheme.
\begin{figure}[htp]
\centering
  \includegraphics[width=3.5in]{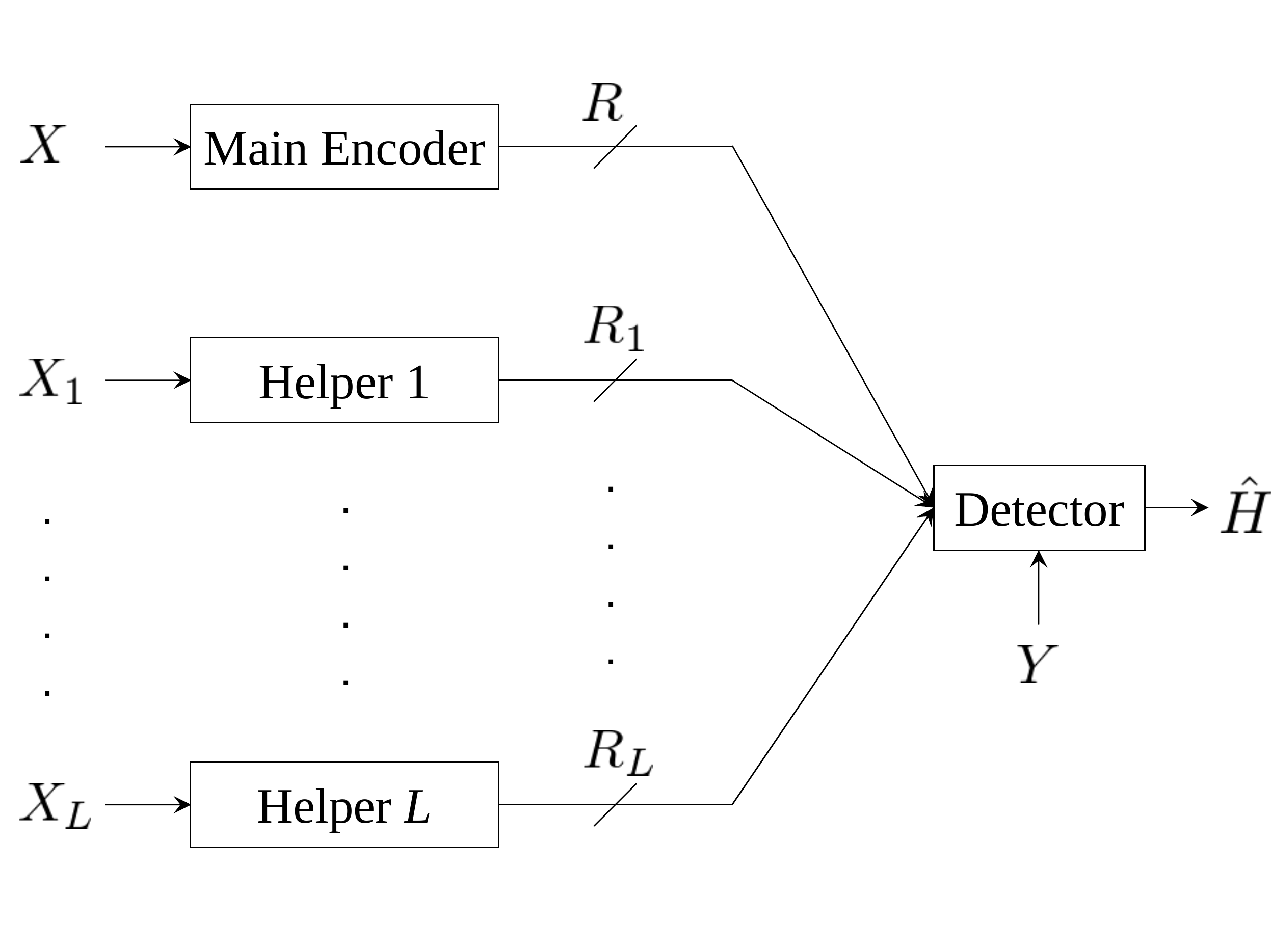}
\caption{Gaussian many-help-one hypothesis testing against independence}\label{fig:Fig4}
\end{figure}

For all three examples, we obtain the solution by observing that the relevant error exponent takes the form of a mutual information, and thereby relate the problem to a source-coding problem. This correspondence was first observed by Ahlswede and Csisz\'{a}r \cite{Ahl}. Tian and Chen later applied it in the context of successive refinement~\cite{Tian}. These three conclusive results enable us to answer both of the above questions. Because the Shimokawa-Han-Amari scheme is optimal for $L=1$, the tradeoff that it entails, depicted in Fig. \ref{fig:Fig2}, must be fundamental to the problem. Moreover, as both the Shimokawa-Han-Amari and Quantize-Bin-Test schemes do not use structured binning, we conclude that it is not necessary for this problem, at least in the special case considered here.

As a byproduct of our results, we obtain an outer bound for a more general class of instances of the distributed hypothesis testing problem. This is the first nontrivial outer bound for the problem, and numerical experiments show that it is quite close to the existing achievable regions in many cases.

The rest of the paper is organized as follows. In Section 2, we introduce the notation used in the paper. We give the mathematical formulation of the $L$-encoder general hypothesis testing problem in Section 3. Section 4 is devoted to the $L$-encoder hypothesis testing against conditional independence problem. Section 5 is on the special case in which there is only one encoder. The Gel`fand and Pinsker hypothesis testing against independence problem is studied in Section 6. The Gaussian many-help-one hypothesis testing against independence problem is studied on Section 7. Finally, we present the outer bound for a class of the general problem in Section 8.

\section{Notation}
We use upper case to denote random variables and vectors. Boldface is used to distinguish vectors from scalars. Arbitrary realizations of random variables and vectors are denoted in lower case. For a random variable $X$, $X^n$ denotes an i.i.d. vector of length $n$, $X^n(i)$ denotes its \emph{i}th component, $X^n(i : j)$ denotes the \emph{i}th through \emph{j}th components, and $X^n(i^c)$ denotes all but the \emph{i}th component. For random variables $X$ and $Y$, we use $\sigma^2_X$ and $\sigma^2_{X|Y}$ to denote the variance of $X$ and the conditional variance of $X$ given $Y$, respectively. The closure of a set $\mathcal{A}$ is denoted by $\overline{\mathcal{A}}$. $|f|$ denotes the cardinality of the range of a function $f$. $1_A$ denotes the indicator function of an event $A$. The determinant of a matrix $\mb{K}$ is denoted by $\det (\mb{K})$. The notation $x^{+}$ denotes $\max(x, 0)$. All logarithms are to the base 2. $\mathbb{R}_{+}^L$ is used to denote the positive orthant in $L$-dimensional Euclidean space. The notation $X \leftrightarrow Y \leftrightarrow Z$ means that $X, Y,$ and $Z$ form a Markov chain in this order. For $0 \le p \le 1$, $H_b(p)$ denotes the binary entropy function defined as
\[
H_b(p) \triangleq - p \log p - (1-p) \log (1-p).
\]
All entropy and mutual information quantities are under the null hypothesis, $H_0$, unless otherwise stated.
\section{$L$-Encoder General Hypothesis Testing}
\subsection{Problem Formulation}
Let $\lsb X_1, \dots, X_L,Y\rsb$ be a generic source taking values in $\prod_{l=1}^{L} \mathcal{X}_l \times \mathcal{Y}$, where $\mathcal{X}_1, \dots, \mathcal{X}_L, $ and $\mathcal{Y}$ are alphabet sets of ${X}_1, \dots, {X}_L, $ and ${Y}$, respectively. The distribution of the source is $P_{{X_1\dots X_LY}}$ under the null hypothesis $H_0$ and is $Q_{{X_1\dots X_LY}}$ under the alternate hypothesis $H_1$, i.e.,
\begin{align*}
H_0: \hspace{0.05in}& P_{{X_1\dots X_LY}} \\
H_1: \hspace{0.05in}& Q_{{X_1\dots X_LY}}.
\end{align*}
Let $\lmb \lsb{X}^n_1(i), \dots, {X}^n_L(i),{Y}^n(i)\rsb \rmb_{i=1}^{n}$ be an i.i.d. sequence of random vectors with the distribution at a single stage same as that of $(X_{1}, \dots, X_{L}, Y)$. We use $\mathcal{L}$ to denote the set $\{1,\dots,L\}$. For $S \subseteq \mathcal{L}$, $S^c$ denotes the complement set $\mathcal{L} \setminus S$ and $\mb{X}^n_{S}(i)$ denotes $({X}^n_l(i))_{l \in S}$. When $S = \mathcal{L}$, we simply write $\mb{X}^n_{\mathcal{L}}(i)$ as $\mb{X}^n(i)$. Likewise when $S = \{l\}$, we write $\mb{X}^n_{\{l\}}(i)$ and $\mb{X}^n_{\{l\}^c}(i)$ as $X^n_l(i)$ and $\mb{X}^n_{l^c}(i)$, respectively. Similar notation will be used for other collections of random variables.

As depicted in Fig. \ref{fig:Fig1}, the encoder $l$ observes $X^n_l$, then sends a message to the detector using an encoding function
\begin{align*}
f_{l}^{(n)} : \mathcal{X}_l^{n} \mapsto \left \{1,\dots,M_{l}^{(n)} \right \}.
\end{align*}
$Y^n$ is available at the detector, which uses it and the messages from the encoders to make a decision between the hypotheses based on a decision rule
\[
g^{(n)}\lsb m_1, \dots, m_L,{y}^n \rsb = \left\{
\begin{array}{l l}
  H_0 & \quad \mbox{if $\lsb m_1, \dots, m_L,{y}^n \rsb \mbox{is in } A$}\\
  H_1 & \quad \mbox{otherwise,}\\ \end{array} \right.
\]
where
\[
A \subseteq \prod_{l=1}^L \lmb 1,\dots,M_{l}^{(n)} \rmb  \times \mathcal{Y}^{n}
\]
is the acceptance region for $H_0$. The encoders $f_l^{(n)}$ and the detector $g^{(n)}$ are such that the type 1 error probability does not exceed a fixed $\epsilon$ in $(0,1)$, i.e.,
\[
P_{\left (f_l^{(n)} \left ({X}_l^n \right)\right)_{l \in \mathcal{L}}{Y}^n}(A^c) \leq \epsilon, \\
\]
and the type 2 error probability does not exceed $\eta$, i.e.,
\[
Q_{\left (f_l^{(n)} \left ({X}_l^n \right)\right)_{l \in \mathcal{L}} {Y}^n}(A) \leq \eta.
\]

\begin{Def}
A rate-exponent vector
\[
(\mb{R},E) = (R_1,\dots,R_L,E)
\]
is \emph{achievable} for a fixed $\epsilon$ if for any positive $\delta$ and sufficiently large $n$, there exists encoders $f^{(n)}_l$ and a detector $g^{(n)}$ such that
\begin{align*}
\frac{1}{n}\log M^{(n)}_l &\leq R_l + \delta \hspace {0.05in} \textrm{for all} \hspace {0.05in} l  \hspace {0.05in} \textrm{in} \hspace {0.05in} \mathcal{L}, \hspace {0.05in} \textrm{and}\\
-\frac{1}{n}\log \eta &\geq E- \delta.
\end{align*}
Let $\mathcal{R}_{\epsilon}$ be the set of all achievable rate-exponent vectors for a fixed $\epsilon$. The \emph{rate-exponent region} $\mathcal{R}$ is defined as
\begin{align*}
\mathcal{R} \triangleq \bigcap_{\epsilon > 0} \mathcal{R}_{\epsilon}.
\end{align*}
\end{Def}

Our goal is to characterize the region $\mathcal{R}$.

\subsection{Entropy Characterization of the Rate-Exponent Region}
We start with the entropy characterization of the rate-exponent region. We shall use it later in the paper to obtain inner and outer bounds. Define the set
\begin{align*}
\mathcal{R}_{*} \triangleq \bigcup_n \bigcup_{\left(f_l^{(n)}\right)_{l \in \mathcal{L}}} \mathcal{R}_{*}\left(n,\left(f_l^{(n)}\right)_{l \in \mathcal{L}}\right),
\end{align*}
where
\begin{align}
\mathcal{R}_{*}\left(n,\left(f_l^{(n)}\right)_{l \in \mathcal{L}}\right) \triangleq \Biggr \{ (\mb{R},E) :
R_l &\ge \frac{1}{n} \log \left|f_l^{(n)}\left({X}_l^n\right)\right| \hspace {0.05 in} \textrm{for all \emph{l} in} \hspace{0.05in} \mathcal{L}, \hspace {0.05 in} \textrm{and} \nonumber\\
E &\le \frac{1}{n} D\biggr(P_{\left (f_l^{(n)} \left ({X}_l^n \right)\right)_{l \in \mathcal{L}}{Y}^n} \Bigr\| Q_{\left (f_l^{(n)} \left ({X}_l^n \right)\right)_{l \in \mathcal{L}}{Y}^n}\biggr) \Biggr \}.
\end{align}
We have the following Proposition.
\begin{Prop}
$\mathcal{R} = \overline{\mathcal{R}_{*}}$.
\end{Prop}
The proof of Proposition 1 is a straight-forward generalization of that of Theorem 1 in \cite{Ahl} and is hence omitted. Ahlswede and Csisz\'{a}r \cite{Ahl} showed that for $L=1$, the strong converse holds, i.e., $\mathcal{R}_{\epsilon}$ is independent of ${\epsilon}$. Thus, $\overline{\mathcal{R}_{*}}$ is essentially a characterization for both $\mathcal{R}$ and $\mathcal{R}_{\epsilon}$. While we expect this to hold for the problem under investigation too, we shall not investigate it here. We next study a class of instances of the problem before returning to the general problem in Section 8.

\section{$L$-Encoder Hypothesis Testing against Conditional Independence}
We consider a class of instances of the general problem, referred to as the $L$-encoder hypothesis testing against conditional independence problem, and obtain inner and outer bounds to the rate-exponent region. These bounds coincide and characterize the region completely in some cases. Moreover, the outer bound for this problem can be used to give an outer bound for a more general class of problems, as we shall see later.

Let $X_{L+1}$ and $Z$ be two generic sources taking values in alphabet sets $\mathcal{X}_{L+1}$ and $\mathcal{Z}$, respectively such that $(\mb{X},X_{L+1})$ and $Y$ are conditionally independent given $Z$ under $H_1$, and the distributions of $(\mb{X},X_{L+1},Z)$ and $(Y,Z)$ are the same under both hypotheses, i.e.,
\begin{align*}
H_0: \hspace{0.05in}& P_{{\mb{X}X_{L+1}Y|Z}}P_Z \\
H_1: \hspace{0.05in}& P_{{\mb{X}X_{L+1}|Z}}P_{Y|Z}P_Z.
\end{align*}
The problem formulation is the same as before with $Y$ replaced by $(X_{L+1},Z,Y)$ in it. The reason for focusing on this special case is that the relative entropy in (1) becomes a mutual information, which simplifies the analysis. Let $\mathcal{R}^{CI}$ be the rate-exponent region of this problem. Here ``$\emph{CI}$" stands for conditional independence. Let
\begin{align*}
\mathcal{R}_{*}^{CI} \triangleq \bigcup_n \bigcup_{\left(f_l^{(n)}\right)_{l \in \mathcal{L}}} \mathcal{R}_{*}^{CI}\left(n,\left(f_l^{(n)}\right)_{l \in \mathcal{L}}\right),
\end{align*}
where
\begin{align*}
\mathcal{R}_{*}^{CI}\left(n,\left(f_l^{(n)}\right)_{l \in \mathcal{L}}\right) \triangleq \Biggr \{ (\mb{R},E) :
R_l &\ge \frac{1}{n} \log \left|f_l^{(n)}\left({X}_l^n\right)\right| \hspace{0.05in} \textrm{for all \emph{l} in} \hspace{0.05in} \mathcal{L}, \hspace {0.05 in} \textrm{and} \\
E &\le \frac{1}{n} I\biggr(\left (f_l^{(n)} \left ({X}_l^n \right)\right)_{l \in \mathcal{L}}, {X}_{L+1}^n;{Y}^n\Bigr|Z^n\biggr) \Biggr \}.\nonumber
\end{align*}
We have the following corollary as a consequence of Proposition 1.
\begin{Cor}
$\mathcal{R}^{CI} = \overline{\mathcal{R}_{*}^{CI}}$.
\end{Cor}
With mutual information replacing relative entropy, the problem can be analyzed using techniques from distributed rate-distortion. In particular, both inner and outer bounds for that problem can be applied here.
\subsection{Quantize-Bin-Test Inner Bound}
Our inner bound is based on a simple scheme which we call the Quantize-Bin-Test scheme. In this scheme, encoders, as in the Shimokawa-Han-Amari scheme, quantize and then bin their observations, but the detector now performs the test directly using the bins. The inner bound obtained is similar to the generalized Berger-Tung inner bound for distributed source coding \cite{Berger1, Tung, Gastper}. Let $\Lambda_i$ be the set of finite-alphabet random variables $\lambda_i = (U_1,\dots, U_L, T)$ satisfying
\begin{enumerate}
\item[(C1)] $T$ is independent of $(\mb{X},X_{L+1},Y,Z)$, and
\item[(C2)] ${U}_l \leftrightarrow ({X}_l,T) \leftrightarrow (\mb{U}_{l^c}, \mb{X}_{l^c},X_{L+1},Y,Z)$ for all $l$ in $\mathcal{L}$.
\end{enumerate}
Define the set
\begin{align*}
\mathcal{R}_i^{CI}(\lambda_i) \triangleq \biggr \{ (\mb{R},E) :
\sum_{l \in S}R_l &\ge I\left(\mb{X}_S; \mb{U}_S|\mb{U}_{S^c},X_{L+1},Z,T \right)\hspace{0.05in} \textrm{for all} \hspace{0.05in}  S \subseteq \mathcal{L}, \hspace {0.05 in} \textrm{and} \\
E &\le I \left(Y;\mb{U},X_{L+1}|Z,T \right) \biggr \}\nonumber
\end{align*}
and let
\begin{align*}
\mathcal{R}_i^{CI} \triangleq \bigcup_{\lambda_i \in \Lambda_i} \mathcal{R}_i^{CI}(\lambda_i).\nonumber
\end{align*}
The following lemma asserts that $\mathcal{R}^{CI}_i$ is computable and closed.
\begin{Lem}
\begin{enumerate}
\item[(a)] $\mathcal{R}_i^{CI}$ remains unchanged if we impose the following cardinality bound on $(\mb{U},T)$ in $\Lambda_i$
\begin{align*}
|\mathcal{U}_l| &\le |\mathcal{X}_l| + 2^{L} - 1 \hspace{0.05in}\textrm{for all l in} \hspace{0.05in} \mathcal{L}, \hspace {0.05 in} \textrm{and} \\
|\mathcal{T}| &\le 2^{L}.
\end{align*}
\item[(b)] $\mathcal{R}_i^{CI}$ is closed.
\end{enumerate}
\end{Lem}
The proof of Lemma 1 is presented in Appendix A. Although the cardinality bound is exponential in the number of encoders, one can obtain an improved bound by exploiting the contra-polymatroid structure of $\mathcal{R}_i^{CI}$ \cite{Chen,Viswanath}. We do not do so here because it is technically involved and we just want to prove that $\mathcal{R}_i^{CI}$ is closed. The following theorem gives an inner bound to the rate-exponent region.
\begin{Thm}
$\mathcal{R}_i^{CI} \subseteq \mathcal{R}^{CI}.$
\end{Thm}
Theorem 1 is proved in Appendix B.

\emph{Remark 1:} Although our inner bound is stated for the special case of the test against conditional independence, it can be extended to the general case. But, the inner bound thus obtained will be quite complicated, with competing exponents, and it is not needed in this paper.

It is worth pointing out that the Quantize-Bin-Test scheme is in general suboptimal for problems in which encoders' observations have common randomness, i.e., there exists deterministic functions of encoders' observations that is common to encoders. However, it is straightforward to generalize this scheme by using the idea from the common-component scheme for distributed source coding problems \cite{Wagner4}.

\subsection{Outer Bound}
The outer bound is similar to the outer bound for the distributed rate-distortion problem given by Wagner and Anantharam \cite{Wagner3}. Let $\Lambda_o$ be the set of finite-alphabet random variables $\lambda_o = (\mb{U}, W, T)$ satisfying
\begin{enumerate}
\item[(C3)] $(W,T)$ is independent of $(\mb{X},X_{L+1},Y,Z)$, and
\item[(C4)] ${U}_l \leftrightarrow ({X}_l,W,T) \leftrightarrow (\mb{U}_{l^c}, \mb{X}_{l^c},X_{L+1},Y,Z)$ for all $l$ in $\mathcal{L}$,
\end{enumerate}
and let $\chi$ be the set of finite-alphabet random variable $X$ such that $X_1, \dots, X_L, X_{L+1}, Y$ are conditionally independent given $(X, Z)$. Note that $\chi$ is nonempty because it contains $(\mb{X},X_{L+1})$. For a given $X$ in $\chi$ and $\lambda_o$ in $\Lambda_o$, the joint distribution of $X$, $(\mb{X},X_{L+1},Y,Z)$, and $\lambda_o$ satisfy the Markov condition
\[
X \leftrightarrow (\mb{X},X_{L+1},Y,Z) \leftrightarrow \lambda_o.
\]
Wagner and Anantharam \cite{Wagner3} refer to this condition as ``{Markov coupling}" between $X$ and $\lambda_o$. Define the set
\begin{align*}
\mathcal{R}_o^{CI}(X,\lambda_o) \triangleq \biggr \{ (\mb{R},E) :
\sum_{l \in S}R_l &\ge I \left({X}; \mb{U}_S|\mb{U}_{S^c},X_{L+1},Z,T\right) + \sum_{l \in S} I\left(X_l ; U_l | X, W,X_{L+1}, Z,T\right) \hspace{0.05in} \textrm{for all} \hspace{0.05in}  S \subseteq \mathcal{L}, \hspace {0.05 in} \textrm{and} \\
E &\le I\left(Y;\mb{U},X_{L+1}|Z,T\right) \biggr \}.\nonumber
\end{align*}
Also let
\begin{align*}
\mathcal{R}_o^{CI} \triangleq \bigcap_{X \in \chi} \bigcup_{\lambda_o \in \Lambda_o} \mathcal{R}_o^{CI}(X,\lambda_o).\nonumber
\end{align*}
We have the following outer bound to the rate-exponent region.
\begin{Thm}
$\mathcal{R}_{*}^{CI} \subseteq \mathcal{R}_o^{CI}$ and therefore $\mathcal{R}^{CI} \subseteq \overline{\mathcal{R}_o^{CI}}.$
\end{Thm}
The proof of the first inclusion is presented in Appendix C. The first inclusion and Corollary 1 imply the second inclusion. The next three sections provide examples in which the inner and outer bounds coincide. In Section 8, we will see how to extend the outer bound to the general problem.

\section{$1$-Encoder Hypothesis Testing against Conditional Independence}
In this section, we study a special case in which $L=1$. We prove that the Quantize-Bin-Test scheme is optimal for this problem. We also prove that the Shimokawa-Han-Amari inner bound coincides with the Quantize-Bin-Test inner bound, establishing that the Shimokawa-Han-Amari scheme is also optimal.

\subsection{Rate-Exponent Region}
\begin{Thm}
For this problem, the rate-exponent region
\begin{align}
\mathcal{R}^{CI} =  \overline{\mathcal{R}^{CI}_o} &= \mathcal{R}^{CI}_i \\
&= \tilde {\mathcal{R}}^{CI} \triangleq \Bigr \{(R_1,E):  \hspace{0.05in}\textrm{there exists} \hspace{0.05in}U_1 \hspace{0.05in}\textrm{such that} \nonumber\\
&\hspace{1in}R_1 \ge I(X_1 ;  U_1 | {X}_2, Z), \nonumber\\
&\hspace{1in}E \le I(Y;{U_1},{X}_2|Z), \\
&\hspace{1in}| {\mathcal{U}}_1| \le |\mathcal{X}_1| + 1, \hspace{0.05in}  \textrm{and} \nonumber\\
&\hspace{1in} U_1 \leftrightarrow X_1 \leftrightarrow ({X}_2, Y, Z)\Bigr \}.\nonumber
\end{align}
\end{Thm}
\begin{proof}
To show (2), it suffices to show that
\[
\mathcal{R}^{CI}_o \subseteq \mathcal{R}^{CI}_i,
\]
because $\mathcal{R}^{CI}_i$ is closed from Lemma 1(b). Consider $(R_1,E)$ in $\mathcal{R}^{CI}_o$. Take $X={X}_2$. It is evident that ${X}_2$ is in $\chi$. Then there exists $\lambda_o = (U_1,W,T)$ in $\Lambda_o$ such that $(R_1,E)$ is in $\mathcal{R}^{CI}_o({X}_2,\lambda_o)$, i.e.,
\begin{align*}
R_1 &\ge I({X}_2; {U}_1| {X}_2,Z,T) + I(X_1 ; U_1 | {X}_2, Z,W,T)\\
&=I(X_1 ; U_1 | {X}_2, Z,W,T),
\end{align*}
and
\begin{align}
E &\le I(Y;{U_1},{X}_2|Z,T) \nonumber\\
&=H(Y|Z,T) - H(Y|U_1,{X}_2,Z,T) \nonumber\\
&\le H(Y|Z,W,T) - H(Y|U_1,{X}_2,Z,W,T)\\
&= I(Y;{U_1},{X}_2|Z,W,T),\nonumber
\end{align}
where (4) follows from conditioning reduces entropy and the fact that $(Y,Z)$ is independent of $(W,T)$. If we set $\tilde T = (W,T)$, then it is easy to verify that $\lambda_i = (U_1,\tilde T)$ is in $\Lambda_i$ and we have
\begin{align}
R_1 &\ge I(X_1 ; U_1 | {X}_2, Z,\tilde T), \hspace{0.05in} \textrm{and}\\
E &\le I(Y;{U_1},{X}_2|Z,\tilde T).
\end{align}
Therefore, $(R_1,E)$ is in $\mathcal{R}^{CI}_i(\lambda_i)$, which implies that $(R_1,E)$ is in $\mathcal{R}^{CI}_i$. This completes the proof of (2).

To prove (3), it suffices to show that
\[
\mathcal{R}^{CI}_i \subseteq \tilde {\mathcal{R}}^{CI}.
\]
The reverse containment immediately follows if we restrict $T$ to be deterministic in the definition of $\mathcal{R}^{CI}_i$. Continuing from the proof of (2), let $\tilde U_1 = (U_1, \tilde T)$. Since $(U_1,\tilde T)$ is in $\Lambda_i$, we have that $\tilde T$ is independent of $({X}_1,X_2,Y,Z)$ and that
\[
U_1 \leftrightarrow (\tilde T, X_1) \leftrightarrow ({X}_2,Y,Z).
\]
Both together imply that
\begin{align*}
\tilde U \leftrightarrow X_1 \leftrightarrow ({X}_2,Y,Z).
\end{align*}
We next have from (5) that
\begin{align}
R_1 &\ge I(X_1 ; U_1 | {X}_2, Z,\tilde T) \nonumber\\
&= I(X_1 ; U_1 | {X}_2, Z,\tilde T) + I(X_1 ; \tilde T | {X}_2, Z)\\
&=I(X_1 ; U_1,\tilde T| {X}_2, Z) \nonumber\\
&=I(X_1 ; \tilde U_1| {X}_2, Z),\nonumber
\end{align}
where (7) follows because $\tilde T$ is independent of $({X}_1,X_2,Y,Z)$. And (6) similarly yields
\begin{align*}
E &\le I(Y;\tilde {U}_1,{X}_2|Z).\nonumber
\end{align*}
Using the support lemma \cite[Lemma 3.4, pp. 310]{csiszar} as in the proof of Lemma 1(a), we can obtain the cardinality bound
\[
|\tilde {\mathcal{U}}_1| \le |\mathcal{X}_1| + 1 \nonumber.
\]
We thus conclude that $(R_1,E)$ is in $\tilde {\mathcal{R}}^{CI}$.
\end{proof}

\subsection{Optimality of Shimokawa-Han-Amari Scheme}
 \begin{figure}[htp]
\centering
  \includegraphics[width=3.7in]{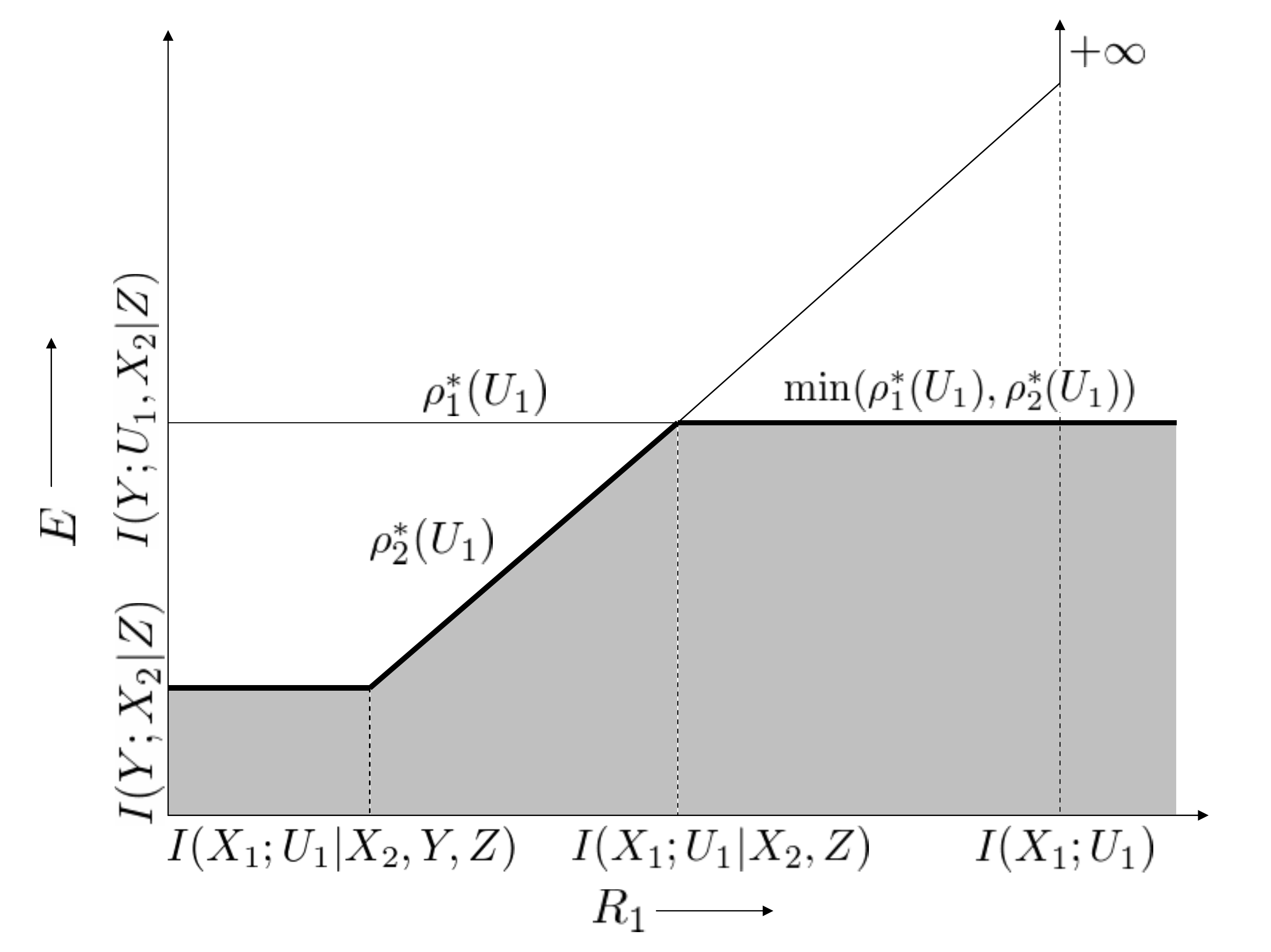}
\caption{Shimokawa-Han-Amari achievable region for a fixed $P_{U_1|X_1}$}\label{fig:Fig5}
\end{figure}
The Shimokawa-Han-Amari scheme operates as follows. Consider a test channel $P_{U_1|X_1}$, a sufficiently large block length $n$, and $\alpha > 0$. Let $\bar R_1 = I(X_1;U_1) + \alpha$. To construct the codebook, we first generate $2^{n \bar R_1}$ independent codewords $U_1^n$, each according to $\prod_{i=1}^n P_{U_1}(u_{1i})$, and then distribute them uniformly into $2^{n R_1}$ bins. The codebook and the bin assignment are revealed to the encoder and the detector. The encoder first quantizes $X_1^n$ by selecting a codeword $U_1^n$ that is jointly typical with it. With high probability, there will be at least one such codeword. The encoder then sends to the detector the index of the bin to which the codeword $U_1^n$ belongs. The joint type of $(X_1^n,U_1^n)$ is also sent to the detector, which requires zero additional rate asymptotically. The detector finds a codeword $\hat {U}_1^n$ in the bin that minimizes the empirical entropy $H({U}^n_1,Y^n)$. It then performs the test and declares $H_0$ if and only if both $(X_1^n,U_1^n)$ and $(Y^n,\hat{U}_1^n)$ are jointly typical under $H_0$. The inner bound thus obtained is as follows. Define
\begin{align*}
\mathcal{A}(R_1) &\triangleq \Bigr \{U_1: R_1 \ge I(X_1;U_1|{X}_2,Y,Z), \hspace{0.1in} U_1 \leftrightarrow X_1 \leftrightarrow ({X}_2,Y,Z), \hspace {0.05 in} \textrm{and} \hspace {0.05 in} |\mathcal{U}_1| \le |\mathcal{X}_1| + 1 \Bigr \} \\
\mathcal{B}(U_1) &\triangleq \Bigr \{P_{\tilde U_1 \tilde X_1 \tilde X_2 \tilde Y \tilde Z}: P_{\tilde U_1 \tilde {X}_1} = P_{U_1X_1}\hspace {0.05 in} \textrm{and}\hspace {0.05 in} P_{\tilde U \tilde {{X}}_{2} \tilde Y \tilde Z } = P_{U_1 {X}_2YZ} \Bigr\}\\
\mathcal{C}(U_1) &\triangleq \Bigr \{P_{\tilde U_1 \tilde X_1 \tilde X_2 \tilde Y \tilde Z }: P_{\tilde U_1 \tilde {X}_1} = P_{U_1X_1}, \hspace {0.05 in} P_{\tilde {{X}}_{2}  \tilde Y \tilde Z} = P_{{X}_2YZ}, \hspace {0.05 in} \textrm{and}\hspace {0.05 in} H (\tilde U_1 | \tilde {{X}}_{2},  \tilde Y,\tilde Z ) \ge H\left (U_1|{X}_2,Y,Z\right) \Bigr\}.
\end{align*}
In addition, define the exponents
\begin{align*}
\rho_1^{*}(U_1) &\triangleq \min_{P_{\tilde U_1 \tilde X_1 \tilde X_2 \tilde Y \tilde Z}  \in \mathcal{B}(U_1)} D \bigr(P_{\tilde U_1\tilde X_1 \tilde X_2 \tilde Y \tilde Z } \|P_{U_1|X_1} P_{X_1 X_2|Z}P_{Y|Z}P_Z\bigr) \\
\rho_2^{*}(U_1) &\triangleq  \left\{
\begin{array}{l l}
  + \infty & \quad \mbox{if $R_1 \ge I(U_1;X_1)$}\\
  \rho_2(U_1) & \quad \mbox{otherwise}\\ \end{array} \right. \\
\rho_2(U_1) &\triangleq [R_1 - I(X_1;U_1|{X}_2,Y,Z)]^{+} + \min_{P_{\tilde U_1 \tilde X_1 \tilde X_2 \tilde Y \tilde Z} \in \mathcal{C}(U_1)} D \bigr(P_{\tilde U_1 \tilde X_1 \tilde X_2 \tilde Y \tilde Z } \|P_{U_1|X_1} P_{X_1 X_2|Z}P_{Y|Z}P_Z\bigr).
\end{align*}
Finally, define
\begin{align*}
E_{SHA}(R_1) \triangleq \max_{U_1 \in \mathcal{A}(R_1)} \min \hspace {0.05 in} \left ( \rho_1^{*}(U_1), \rho_2^{*}(U_1) \right).
\end{align*}
Recall that $\rho_2^{*}(U_1)$ and $\rho_1^{*}(U_1)$ are the exponents associated with type 2 errors due to binning errors and assuming correct decoding of the codeword, respectively.
\begin{Thm}
\cite{Han2} $(R_1,E)$ is in the rate-exponent region if
\[
E \le E_{SHA}(R_1).
\]
\end{Thm}

Fig. \ref{fig:Fig5} shows the Shimokawa-Han-Amari achievable exponent as a function of the rate assuming a fixed channel $P_{U_1|X_1}$ is used for quantization. This is simply Fig. \ref{fig:Fig2} particularized to the $1$-encoder hypothesis testing against conditional independence problem. For rates $R_1 \ge I(X_1,U_1|{X}_2,Z)$, $\rho_1^{*}(U_1)$ dominates $\rho_2^{*}(U_1)$ and there is no penalty for binning at these rates as the exponent stays the same. Therefore, we can bin all the way down to the rate $R_1 = I(X_1,U_1|{X}_2,Z)$ without any loss in the exponent. However, if we bin further at rates $R_1$ in $[I(X_1,U_1|{X}_2,Y,Z),I(X_1,U_1|{X}_2,Z))$, then $\rho_2^{*}(U_1)$ dominates $\rho_1^{*}(U_1)$, the exponent decreases linearly with $R_1$, and the performance deteriorates all the way down to a point at which the message from the encoder is useless. At this point, the binning rate $R_1$ equals $I(X_1,U_1|{X}_2,Y,Z)$ and the exponent equals $I(Y;{X}_2|Z)$, which is the exponent when the detector ignores the encoder's message. This competition between the exponents makes the optimality of the Shimokawa-Han-Amari scheme unclear. We prove that it is indeed optimal by showing that the Shimokawa-Han-Amari inner bound simplifies to the Quantize-Bin-Test inner bound, which by Theorem 3 is tight. Let us define
\[
\mathcal{A}^{*}(R_1) \triangleq \Bigr \{U_1: R_1 \ge I(X_1;U_1|{X}_2,Z), \hspace{0.1in} U_1 \leftrightarrow X_1 \leftrightarrow ({X}_2,Y,Z), \hspace {0.05 in} \textrm{and}\hspace {0.05 in}|\mathcal{U}_1| \le |\mathcal{X}_1| + 1 \Bigr \}
\]
and
\[
E_{QBT}(R_1) \triangleq \max_{U_1 \in \mathcal{A}^{*}(R_1)} I(Y;{U_1},{X}_2|Z).
\]
We have the following theorem.
\begin{Thm}
If $(R_1,E)$ is in the rate-exponent region, then
\[
E \le E_{QBT}(R_1) = E_{SHA}(R_1).
\]
\end{Thm}
\begin{proof}
The inequality follows from Theorem 3. To prove the equality, it is sufficient to show that
\[
E_{SHA}(R_1) \ge E_{QBT}(R_1).
\]
The reverse inequality follows from Theorem 3 and 4. Since conditioning reduces entropy and any $U_1$ in $\mathcal{A}^{*}(R_1)$ satisfies the Markov chain
\[
U_1 \leftrightarrow X_1 \leftrightarrow ({X}_2,Y,Z),
\]
we have
\begin{align*}
R_1 &\ge I(X_1;U_1|{X}_2,Z) \\
&= H(U_1|{X}_2,Z) - H(U_1|{X}_1{X}_2,Z)\\
&\ge H(U_1|{X}_2,Y,Z) - H(U_1|{X}_1{X}_2,Y,Z)\\
&= I(X_1;U_1|{X}_2,Y,Z),
\end{align*}
which means that $U_1$ is in $\mathcal{A}(R_1)$. Hence, $\mathcal{A}^{*}(R_1) \subseteq \mathcal{A}(R_1)$. This implies that
\begin{align}
E_{SHA}(R_1) &\triangleq \max_{U_1 \in \mathcal{A}(R_1)} \min \hspace {0.05 in} \left ( \rho_1^{*}(U_1), \rho_2^{*}(U_1) \right) \nonumber\\
&\ge \max_{U_1 \in \mathcal{A}^{*}(R_1)} \min \hspace {0.05 in} \left ( \rho_1^{*}(U_1), \rho_2^{*}(U_1) \right).
\end{align}
Now the objective of the optimization problem in the definition of $\rho_1^{*}(U_1)$ can be lower bounded as
\begin{align*}
D \bigr(P_{\tilde U_1\tilde {{X}}_1 \tilde {{X}}_2 \tilde Y \tilde Z } \|P_{U_1|X_1} P_{{{X}}_1 {{X}}_2|Z}P_{Y|Z}P_Z\bigr) &\ge D \bigr(P_{\tilde U_1 \tilde {{X}}_{2} \tilde Y \tilde Z } \|P_{U_1{X}_2|Z}P_{Y|Z}P_Z\bigr) \\
&= D \bigr(P_{ U_1 {{X}}_{2} Y Z } \|P_{U_1{X}_2|Z}P_{Y|Z}P_Z\bigr) \\
&= I(Y; {U}_1,{X}_2|Z).
\end{align*}
The lower bound is achieved by the distribution $P_{ U_1 {{X}}_{2} Y Z} P_{X_1|U_1 {{X}}_{2} Z}$ in $\mathcal{B}(U_1)$. Therefore,
\begin{align*}
\rho_1^{*}(U_1) = I(Y; {U}_1,{X}_2|Z).
\end{align*}
Similarly, we can lower bound the optimization problem in the definition of $\rho_2(U_1)$ as
\begin{align*}
D \bigr(P_{\tilde U_1\tilde {{X}}_1 \tilde {{X}}_2 \tilde Y \tilde Z } \|P_{U_1|X_1} P_{{X}_1 X_2|Z}P_{Y|Z}P_Z\bigr) &\ge D \bigr(P_{ \tilde {{X}}_{2} \tilde Y \tilde Z } \|P_{{X}_2|Z}P_{Y|Z}P_Z\bigr) \\
&= D \bigr(P_{ {{X}}_{2} Y Z } \|P_{{X}_2|Z}P_{Y|Z}P_Z\bigr) \\
&= I(Y; {X}_2|Z),
\end{align*}
and the lower bound is achieved by the distribution $P_{{{X}}_{2} Y Z} P_{U_1 X_1| {{X}}_{2} Z}$ in $\mathcal{C}(U_1)$. Therefore,
\begin{align*}
\rho_2(U_1) = [R_1 - I(X_1;U_1|{X}_2,Y,Z)]^{+} + I(Y; {X}_2|Z).
\end{align*}
Consider any $U_1$ in $\mathcal{A}^{*}(R_1).$ If $R_1 \ge I(X_1;U_1)$, then
\begin{align}
\min \hspace {0.05 in} \left ( \rho_1^{*}(U_1), \rho_2^{*}(U_1) \right) &=  \rho_1^{*}(U_1) \nonumber\\
&= I(Y; {U}_1,{X}_2|Z).
\end{align}
And if $I(X_1;U_1) > R_1 \ge I(X_1;U_1|{X}_2,Z)$, then
\begin{align}
\min \hspace {0.05 in} \left ( \rho_1^{*}(U_1), \rho_2^{*}(U_1) \right) &=  \min \hspace {0.05 in} \bigr  ( I(Y; {U}_1,{X}_2|Z), R_1 - I(X_1;U_1|{X}_2,Y,Z) + I(Y; {X}_2|Z) \bigr ) \nonumber\\
&\ge \min \hspace {0.05 in} \bigr ( I(Y; {U}_1,{X}_2|Z), I(X_1;U_1|{X}_2,Z) - I(X_1;U_1|{X}_2,Y,Z) + I(Y; {X}_2|Z) \bigr)\nonumber\\
&= \min \hspace {0.05 in} \bigr ( I(Y; {U}_1,{X}_2|Z), I(Y;U_1|{X}_2,Z)+ I(Y; {X}_2|Z) \bigr)\nonumber\\
&= \min \hspace {0.05 in} \bigr ( I(Y; {U}_1,{X}_2|Z), I(Y; {U}_1,{X}_2|Z) \bigr)\nonumber\\
&= I(Y; {U}_1,{X}_2|Z).
\end{align}
Now (8) through (10) imply
\begin{align}
E_{SHA}(R_1) &\ge \max_{U_1 \in \mathcal{A}^{*}(R_1)} I(Y; {U}_1,{X}_2|Z) \nonumber\\
&= E_{QBT} (R_1).\nonumber
\end{align}
Theorem 5 is thus proved.
\end{proof}

\section{Gel`fand and Pinsker Hypothesis Testing against Independence}
We now consider another special case, which we call the Gel`fand and Pinsker hypothesis testing against independence problem, because it is related to the source coding problem studied by Gel`fand and Pinsker \cite{Gelfand}.

Suppose that $X_{L+1}$ and $Z$ are deterministic and suppose there exists
a function of $X_1$, \ldots, $X_L$, say $X$, such that under $H_0$,
\begin{enumerate}
\item[(C5)] $ X_1,..,X_L, Y$ are conditionally independent given $X$, and
\item[(C6)] for any finite-alphabet random variable $U$ such that $Y \leftrightarrow X \leftrightarrow U$ and $Y \leftrightarrow U \leftrightarrow X$, we have $H(X|U) = 0.$
\end{enumerate}
Conditions (C5) and (C6) imply that under $H_0$, $X$ is a minimal sufficient statistic for $Y$ given $\mb{X}$ such that $ X_1,\dots,X_L, Y$ are conditionally independent given $X$. We shall characterize the centralized rate region, the set of rate vectors that achieve the centralized type 2 error exponent $I(\mb{X};Y) = I(X;Y)$. More precisely, we shall characterize the set
\[
\left \{\mb{R} : (\mb{R},I(X;Y)) \in \mathcal{R}^{CI} \right\},
\]
denoted by $\mathcal{R}^{CI}\bigr(I(X;Y)\bigr)$. We define $\mathcal{R}^{CI}_i\bigr(I(X;Y)\bigr)$ and $\overline{\mathcal{R}^{CI}_o}\bigr(I(X;Y)\bigr)$ similarly. We need the following lemma.
\begin{Lem}
Condition (C6) is equivalent to
\begin{enumerate}
\item[(C7)] For any positive $\epsilon$, there exists a positive $\delta$ such that for all finite-alphabet random variables $U$ such that $Y \leftrightarrow X \leftrightarrow U$ and $I(X;Y|U) \le \delta$, we have $H(X|U) \le \epsilon.$
\end{enumerate}
\end{Lem}
The proof of Lemma 2 is presented in Appendix D. Let us define a function
\begin{align*}
\phi(\delta) \triangleq \inf \Bigr \{ \epsilon : \hspace {0.05 in} \textrm{for all finite-alphabet} \hspace {0.05 in} U \hspace {0.05 in} \textrm{such that}\hspace {0.05 in} &Y \leftrightarrow X \leftrightarrow U \hspace {0.05 in} \textrm{and} \hspace {0.05 in}I(X;Y|U) \le \delta, \hspace {0.05 in} \textrm{we have}\hspace {0.05 in} H(X|U) \le \epsilon \Bigr\}.
\end{align*}
It is clear that $\phi$ is continuous at zero with the value $ \phi (0) = 0.$ We have the following theorem.
\begin{Thm}
For this problem, the centralized rate region
\[
\mathcal{R}^{CI}\bigr(I(X;Y)\bigr) = \mathcal{R}^{CI}_i\bigr(I(X;Y)\bigr) = \overline{\mathcal{R}^{CI}_o}\bigr(I(X;Y)\bigr).
\]
\end{Thm}
\begin{proof}
It suffices to show that
\[
\overline{\mathcal{R}^{CI}_o}\bigr(I(X;Y)\bigr) \subseteq \mathcal{R}^{CI}_i \bigr (I(X;Y) \bigr).
\]
Consider any $\mb{R}$ in $\overline{\mathcal{R}^{CI}_o}\bigr(I(X;Y)\bigr)$, any positive $\delta$, and $X$ defined as above. Then there exists $\lambda_o = (\mb{U},W,T)$ in $\Lambda_o$ such that $(R_1+\delta,\dots,R_L+\delta,I(X;Y)-\delta)$ is in $\mathcal{R}^{CI}_o(X,\lambda_o)$, i.e.,
\begin{align}
\sum_{l \in S}(R_l+\delta) &\ge I({X}; \mb{U}_S|\mb{U}_{S^c},T) + \sum_{l \in S} I(X_l ; U_l | X, W, T) \hspace {0.05 in} \textrm{for all} \hspace{0.05in}  S \subseteq \mathcal{L}, \hspace {0.05 in} \textrm{and} \\
I(X;Y) -\delta &\le I(Y;\mb{U}|T).
\end{align}
We have the Markov chain
\begin{align*}
Y \leftrightarrow X \leftrightarrow (\mb{U},T),
\end{align*}
which implies
\begin{align*}
I(X;Y|\mb{U},T) &= H(Y|\mb{U},T) - H (Y|X,\mb{U},T) \\
&= H(Y|\mb{U},T) - H (Y|X) \\
&= I(X;Y) - I(Y;\mb{U}|T)\\
&\le \delta,
\end{align*}
where the last inequality follows from (12). Therefore, by the definition of $\phi$ function
\begin{align}
H(X|\mb{U},T) \le \phi(\delta).
\end{align}
Now
\begin{align}
I(X;\mb{U}_S|\mb{U}_{S^c},T) &= H(X|\mb{U}_{S^c},T) - H(X|\mb{U},T) \nonumber \\
&\ge H(X|\mb{U}_{S^c},W,T) - \phi(\delta) \\
&\ge I(X;\mb{U}_S|\mb{U}_{S^c},W,T) - \phi(\delta),\nonumber
\end{align}
where (14) follows from (13) and the fact that conditioning reduces entropy. This together with (11) implies
\begin{align}
\sum_{l \in S}(R_l+\delta+\phi(\delta)) &\ge I({X}; \mb{U}_S|\mb{U}_{S^c},W,T) + \sum_{l \in S} I(X_l ; U_l | X, W, T) \nonumber\\
&= I({X}; \mb{U}_S|\mb{U}_{S^c},W,T) + I(\mb{X}_S; \mb{U}_S | \mb{U}_{S^c}, X,W, T) \nonumber\\
&=I({X},\mb{X}_S; \mb{U}_S|\mb{U}_{S^c},W,T) \nonumber\\
&\ge I(\mb{X}_S; \mb{U}_S|\mb{U}_{S^c},W,T).\nonumber
\end{align}
Again since conditioning reduces entropy and $Y$ is independent of $(W,T)$, we obtain from (12) that
\begin{align*}
I(X;Y) - \delta &\le I(Y;\mb{U}|T) \\
&= H(Y|T) - H(Y|\mb{U},T) \\
& \le H(Y|W,T) - H(Y|\mb{U},W,T)\\
&= I(Y;\mb{U}|W,T).
\end{align*}
Define $\tilde T = (W,T)$. It is then clear that $\lambda_i = (\mb{U},\tilde T)$ is in $\Lambda_i$,
\begin{align*}
\sum_{l \in S}(R_l+\delta+\phi(\delta)) &\ge I(\mb{X}_S; \mb{U}_S|\mb{U}_{S^c}, \tilde T) \hspace{0.05in} \textrm{for all} \hspace{0.05in}  S \subseteq \mathcal{L},\hspace{0.05in} \textrm{and} \\
I(X;Y) -\delta &\le I(Y;\mb{U}|\tilde T).
\end{align*}
Hence, $(R_1+\delta+\phi(\delta),\dots,R_L+\delta+\phi(\delta),I(X;Y)-\delta)$ is in $\mathcal{R}^{CI}_i(\lambda_i),$ which implies that $(\mb{R},I(X;Y))$ is in $\mathcal{R}^{CI}_i$ because $\mathcal{R}^{CI}_i$ is closed from Lemma 1(b). Therefore, $\mb{R}$ is in $\mathcal{R}^{CI}_i \bigr ( I(X;Y)\bigr).$
\end{proof}

\section{Gaussian Many-Help-One Hypothesis Testing against Independence}
We now turn to a continuous example of the problem studied in Section 4. This problem is related to the quadratic Gaussian many-help-one source coding problem \cite{Wagner,Oohama2005, Vinod}. We first obtain an outer bound similar to the one in Theorem 2 and then show that it is achieved by the Quantize-Bin-Test scheme.

Let $\lsb X, Y, X_1, \dots, X_L\rsb$ be a zero-mean Gaussian random vector such that
\[
X_l = X + N_l
\]
for each $l$ in $\mathcal{L}$. ${X}$ and ${Y}$ are correlated under the null hypothesis $H_0$ and are independent under the alternate hypothesis $H_1$, i.e.,
\begin{align*}
H_0: \hspace{0.05in} &Y = X + N \\
H_1: \hspace{0.05in} &Y \independent X.
\end{align*}
We assume that $X,N, N_1, N_2, \dots, N_L$ are mutually independent, and that $\sigma^2_N$ and $\sigma^2_{N_l}$ are positive. The setup of the problem is shown in Fig. \ref{fig:Fig4}. Unlike the previous problem, we now allow $X$ to be observed by an encoder, which sends a message to the detector at a finite rate $R$. We use $f^{(n)}$ to denote the corresponding encoding function. In order to be consistent with the source coding terminology, we call this the main encoder. The encoder observing $X_l$ is now called helper $l$. We assume that $X_{L+1}$ and $Z$ are deterministic. The rest of the problem formulation is the same as the one in Section 3.1. Let $\mathcal{R}^{MHO}$ be the rate-exponent region of this problem. We need the entropy characterization of $\mathcal{R}^{MHO}$. For that, define
\begin{align*}
\mathcal{R}_{*}^{MHO} \triangleq  \bigcup_n \bigcup_{f^{(n)}, \left(f_l^{(n)}\right)_{l \in \mathcal{L}}} \mathcal{R}_{*}^{MHO}\left(n,\left(f_l^{(n)}\right)_{l \in \mathcal{L}}\right),
\end{align*}
where
\begin{align*}
\mathcal{R}_{*}^{MHO}\left(n,\left(f_l^{(n)}\right)_{l \in \mathcal{L}}\right) \triangleq \Biggr \{ \bigr(R, \mb{R},E \bigr) :
R &\ge \frac{1}{n} \log \left|f^{(n)}\left({X}^n\right)\right|, \nonumber\\
R_l &\ge \frac{1}{n} \log \left|f_l^{(n)}\left({X}_l^n\right)\right| \hspace {0.05 in} \textrm{for all} \hspace {0.05 in} l \hspace {0.05 in} \textrm{in} \hspace {0.05 in}\mathcal{L}, \hspace {0.05 in} \textrm{and} \\
E &\le \frac{1}{n} I\left(Y^n; f^{(n)} (X^n),\left(f_l^{(n)} \left ({X}_l^n \right)\right)_{l \in \mathcal{L}}\right) \Biggr \}.\nonumber
\end{align*}
\begin{Cor}
$\mathcal{R}^{MHO} = \overline{\mathcal{R}_{*}^{MHO}}$.
\end{Cor}
The proof of this result is almost identical to that of Proposition 1. Define the set
\begin{align*}
\tilde{\mathcal{R}}^{MHO} \triangleq \Biggr \{&(R, R_1,\dots,R_L, E) : \hspace{0.05in} \textrm{there exists} \hspace{0.05in} (r_1,\dots,r_L) \in \mathbb{R}_{+}^L \hspace{0.05in} \textrm{such that} \\
&R_l \ge r_l \hspace{0.05in}\textrm{for all} \hspace{0.05in} l \hspace{0.05in} \textrm{in} \hspace{0.05in} \mathcal{L}, \hspace{0.05in} \textrm{and}\\
&R + \sum_{l \in S} R_l \ge \frac{1}{2} \log^{+} \left [ \frac{1}{D } \left (\frac{1}{\sigma^2_X} + \sum_{l \in  S^c} \frac{1- 2^{- 2 r_l}}{\sigma^2_{N_l}} \right )^{-1}\right ] + \sum_{l \in S} r_l \hspace{0.05in} \textrm{for all} \hspace{0.05in} S \subseteq \mathcal{L} \Biggr\},
\end{align*}
where
\[
D = (\sigma^2_X+\sigma^2_N) 2^{-2E} - \sigma^2_N.
\]
\begin{Thm} The rate-exponent region of this problem
\[
\mathcal{R}^{MHO} = \tilde{\mathcal{R}}^{MHO}.
\]
\end{Thm}
\begin{proof}
The proof of inclusion $\mathcal{R}^{MHO} \subseteq \tilde{\mathcal{R}}^{MHO}$ is similar to the converse proof of the Gaussian many-help-one source coding problem by Oohama \cite{Oohama2005} and Prabhakaran \emph{et al.} \cite{Vinod} (see also \cite{Wagner3}). Their proofs continue to work if we replace the original mean square error distortion constraint with the mutual information constraint that we have here. It is noteworthy though that Wang \emph{et al.}'s \cite{wang} approach does not work here because it relies on the distortion constraint.

We start with the continuous extension of Theorem 2. Let $\Lambda_o$ be the set of random variables $\lambda_o = (U,\mb{U}, W, T)$ such that each take values in a finite-dimensional Euclidean space and collectively they satisfy
\begin{enumerate}
\item[(C8)] $(W,T)$ is independent of $(X,\mb{X},Y)$,
\item[(C9)] ${U} \leftrightarrow ({X},W,T) \leftrightarrow (\mb{U}, \mb{X},Y)$,
\item[(C10)] ${U}_l \leftrightarrow ({X}_l,W,T) \leftrightarrow (U,\mb{U}_{l^c}, X,\mb{X}_{l^c},Y)$ for all $l$ in $\mathcal{L}$, and
\item[(C11)] the conditional distribution of $U_l$ given $(W,T)$ is discrete for each $l$.
\end{enumerate}
Define the set
\begin{align}
\mathcal{R}_o^{MHO}(\lambda_o) \triangleq \biggr \{ (R,\mb{R},E) :
R_l &\ge I(X_l ; U_l | X, W, T) \hspace{0.05in} \textrm{for all} \hspace{0.05in}  l \hspace{0.05in} \textrm{in} \hspace{0.05in} \mathcal{L}, \hspace {0.05 in} \\
R+\sum_{l \in S}R_l &\ge I({X}; U,\mb{U}_S|\mb{U}_{S^c},T) + \sum_{l \in S} I(X_l ; U_l | X, W, T) \hspace{0.05in} \textrm{for all} \hspace{0.05in}  S \subseteq \mathcal{L}, \hspace {0.05 in} \textrm{and} \\
E &\le I(Y;U,\mb{U}|T) \biggr \}.
\end{align}
Finally, let
\begin{align*}
\mathcal{R}_o^{MHO} \triangleq \bigcup_{\lambda_o \in \Lambda_o} \mathcal{R}_o^{CI}(\lambda_o).\nonumber
\end{align*}
We have the following lemma.
\begin{Lem}
$\mathcal{R}_{*}^{MHO} \subseteq \mathcal{R}_o^{MHO}.$
\end{Lem}
The inequalities (16) and (17) can be established as in the proof of Theorem 2. In particular, we obtain (16) by considering only those constraints on the sum of rate combinations that include $R$. The inequality (15) is not present in Theorem 2. However, it can be derived easily. We need the following lemma.
\begin{Lem}\cite[Lemma 9]{Wagner3} If $\lambda_o$ is in $\Lambda_o$, then for all $S \subseteq \mathcal{L}$,
\[
2^{2I(X;\mb{U}_S|W,T)} \le 1 + \sum_{l \in S} \frac{1 - 2^{-2I(X_l;U_l|X,W,T)}}{\sigma^2_{N_l}/\sigma^2_X}.
\]
\end{Lem}
Consider any $(R,\mb{R},E)$ in $\mathcal{R}_o^{MHO}$. Then there exists $\lambda_o$ in $\Lambda_o$ such that for all $S \subseteq \mathcal{L}$,
\begin{align}
R+\sum_{l \in S}R_l &\ge I({X}; U,\mb{U}_S|\mb{U}_{S^c},T) + \sum_{l \in S} I(X_l ; U_l | X, W, T) \nonumber\\
&=I({X}; U,\mb{U}|T) - I({X}; \mb{U}_{S^c}|T) + \sum_{l \in S} I(X_l ; U_l | X, W, T),
\end{align}
and
\begin{align}
E &\le I(Y;U,\mb{U}|T).
\end{align}
We can lower bound the first term in (18) by applying the entropy power inequality \cite{Cover} and obtain
\begin{align*}
2^{2 h (Y|U,\mb{U},T)} &= 2^{2 h (X + N|U,\mb{U},T)} \\
&\ge 2^{2 h ({X}|U,\mb{U},T)} + 2^{2 h ({N})} \\
&=  2^{2 h ({X}|U,\mb{U},T)} + 2 \pi e \sigma^2_N,
\end{align*}
which simplifies to
\begin{align}
h (Y|U,\mb{U},T) \ge \frac{1}{2} \log \left (2^{2 h ({X}|U,\mb{U},T)} + 2 \pi e \sigma^2_N \right ).
\end{align}
Now (19) and (20) together imply
\begin{align}
 I(X; U,\mb{U}|T) \ge \frac{1}{2} \log \frac{\sigma^2_X}{(\sigma^2_X+\sigma^2_N)2^{-2E}-\sigma^2_N}.
\end{align}
We next upper bound the second term in (18). Since conditioning reduces entropy and $X$ is independent of $(W,T)$, we have
\begin{align}
I({X}; \mb{U}_{S^c}|T) &= h(X|T) - h(X|\mb{U}_{S^c},T) \nonumber\\
&\le h(X|W,T) - h(X|\mb{U}_{S^c},W,T) \nonumber\\
&=I({X}; \mb{U}_{S^c}|W,T).
\end{align}
Define
\[
r_l \triangleq I(X_l ; U_l | X, W, T).
\]
Then we have from (18), (21), (22), and Lemma 4 that
\begin{align*}
R+\sum_{l \in S}R_l &\ge \frac{1}{2} \log^{+} \left [ \frac{1}{\Bigr ((\sigma^2_X+\sigma^2_N)2^{-2E}-\sigma^2_N \Bigr)} \left (\frac{1}{\sigma^2_X} + \sum_{l \in  S^c} \frac{1- 2^{- 2 r_l}}{\sigma^2_{N_l}} \right )^{-1}\right ] + \sum_{l \in S} r_l.
\end{align*}
On applying Lemma 3 and Corollary 2, we obtain $\mathcal{R}^{MHO} \subseteq \tilde {\mathcal{R}}^{MHO}.$

\begin{figure}[htp]
\centering
  \includegraphics[width=3.5in]{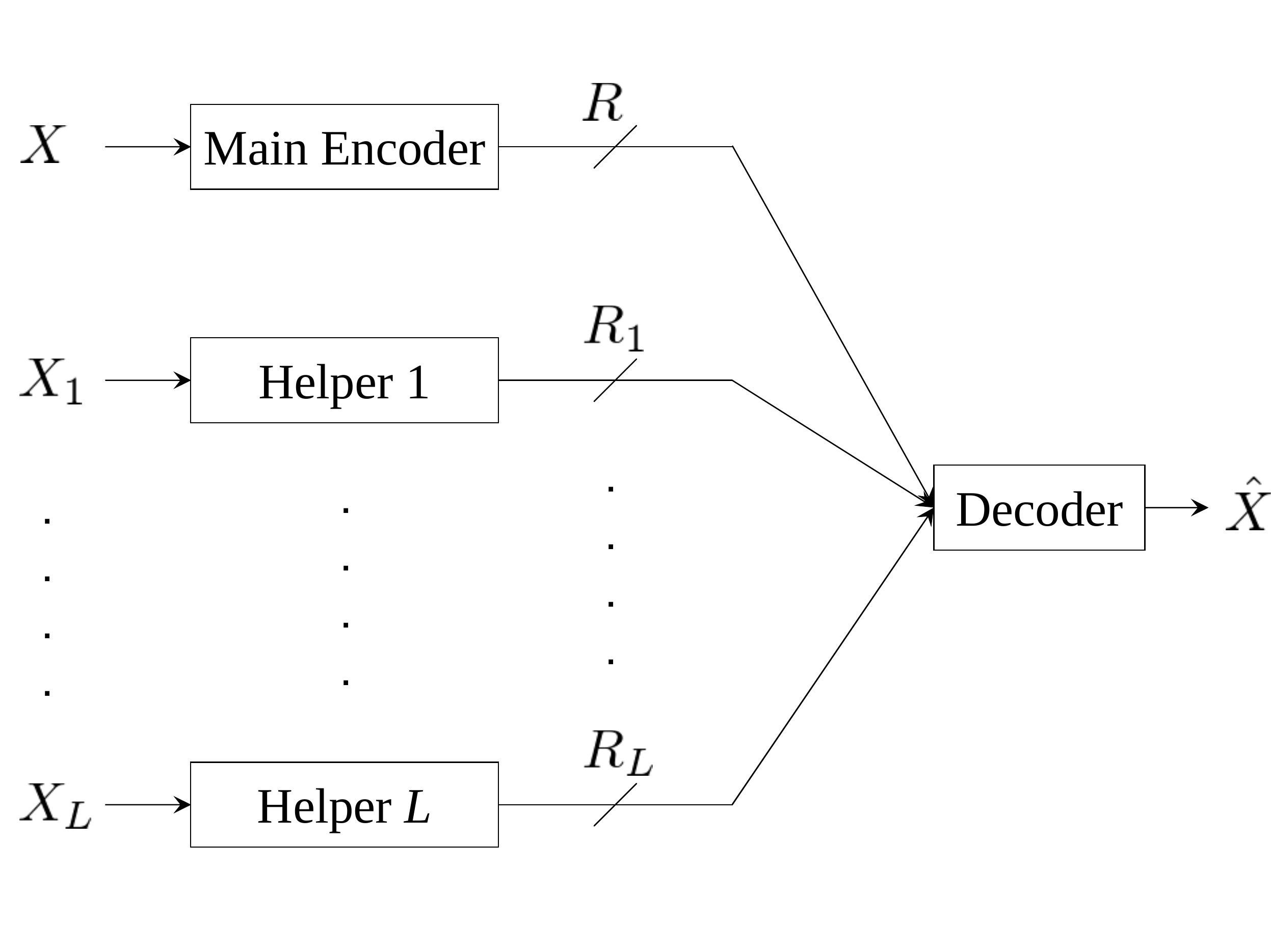}
\caption{Gaussian many-help-one source coding problem}\label{fig:Fig6}
\end{figure}
We use the Quantize-Bin-Test scheme to prove the reverse inclusion. Consider $(R, \mb{R}, E)$ in $\tilde{\mathcal{R}}^{MHO}$. Then there exists $\mb{r} \in \mathbb{R}^L_{+}$ such that
\begin{align*}
&R_l \ge r_l \hspace{0.05in } \textrm{for all \emph{l} in} \hspace{0.05in }\mathcal{L}, \hspace{0.05in } \textrm{and} \\
&R + \sum_{l \in S} R_l  \ge \frac{1}{2} \log^{+} \left [ \frac{1}{D} \left (\frac{1}{\sigma^2_X} + \sum_{l \in S^c} \frac{1- 2^{- 2 r_l}}{\sigma^2_{N_l}} \right )^{-1}\right ] + \sum_{l \in S} r_l \hspace{0.05in } \textrm{for all} \hspace{0.05in } S \subseteq \mathcal{L}.
\end{align*}
We therefore have from Oohama's result \cite{Oohama2005} that  $(R, \mb{R}, D)$ is achievable for the quadratic Gaussian many-help-one source coding problem, the setup of which is shown in Fig. \ref{fig:Fig6}. In this problem, the main encoder and helpers operate as before. The decoder however uses all available information to estimate $X$ such that the mean square error of the estimate is no more than a fixed positive number $D$. Since $(R, \mb{R}, D)$ is achievable, it follows by Oohama's achievability proof that for any positive $\delta$ and sufficiently large $n$, there exists quantize and bin encoders $f^{(n)}, f^{(n)}_1, \dots, f^{(n)}_L$, and a decoder $\psi^{(n)}$ such that
\begin{align}
R+\delta &\ge \frac{1}{n} \log \left |f^{(n)}(X^n) \right|, \\
R_l+\delta &\ge \frac{1}{n} \log \left|f^{(n)}_l(X_l^n)\right| \hspace{0.05in } \textrm{for all \emph{l} in} \hspace{0.05in } \mathcal{L},\hspace{0.05 in} \textrm{and}\\
D+\delta &\ge \frac{1}{n} \sum_{i=1}^n E\left [\left (X^n(i) - \hat X^n(i)\right)^2\right],
\end{align}
where
\begin{align*}
\hat X^n &= \psi^{(n)} \left (f^{(n)}(X^n), \left (f^{(n)}_l(X_l^n) \right)_{l \in \mathcal{L}} \right).
\end{align*}
For each $i$, we have
\begin{align}
E \left [\left(Y^n(i) - \hat X^n(i)\right)^2\right] &= E \left [\left(Y^n(i) - X^n(i) + X^n(i)- \hat X^n(i)\right)^2\right] \nonumber\\
&= E \left [\left(N^n(i) + X^n(i)- \hat X^n(i)\right)^2\right] \nonumber\\
&= \sigma^2_N+E \left [\left(X^n(i)- \hat X^n(i)\right)^2\right],\nonumber
\end{align}
where the last equality follows because
\[
Y^n(i) \leftrightarrow X^n(i) \leftrightarrow \hat X^n(i).
\]
By averaging over time, we obtain
\begin{align}
\frac{1}{n} \sum_{i=1}^n E \left [\left(Y^n(i) - \hat X^n(i)\right)^2\right] &= \sigma^2_N+ \frac{1}{n} \sum_{i=1}^n E \left [\left(X^n(i)- \hat X^n(i)\right)^2\right] \nonumber\\
&\le \sigma^2_N + D + \delta,\nonumber
\end{align}
where the last inequality follows from (25). Therefore, the code achieves a distortion $\sigma^2_N + D + \delta$ in $Y$. Hence,
\[
\frac{1}{n} I \left (Y^n; f^{(n)}(X^n), \left (f^{(n)}_l(X_l^n) \right)_{l \in \mathcal{L}} \right )
\]
must be no less than the rate-distortion function of $Y$ at a distortion $\sigma^2_N + D + \delta$, i.e.,
\begin{align}
\frac{1}{n} I \left (Y^n; f^{(n)}(X^n), \left (f^{(n)}_l(X_l^n) \right)_{l \in \mathcal{L}} \right ) &\ge \frac{1}{2} \log \frac{\sigma^2_X+\sigma^2_N}{\sigma^2_N + D + \delta} \nonumber\\
&=\frac{1}{2} \log \frac{\sigma^2_X+\sigma^2_N}{ (\sigma^2_X+\sigma^2_N) 2^{-2E}+ \delta} \nonumber\\
&\ge\frac{1}{2} \log \frac{\sigma^2_X+\sigma^2_N}{ (\sigma^2_X+\sigma^2_N) 2^{-2(E - \bar {\delta})}} \\
&= E - \bar {\delta},
\end{align}
where (26) follows for a positive $\bar {\delta}$ such that $\bar {\delta} \rightarrow 0$ as ${\delta}\rightarrow 0$. We now have from (23), (24), and (27) that $(R, \mb{R}, E)$ is in $\overline{\mathcal{R}_{*}^{MHO}}$. Hence by Corollary 2, $\tilde {\mathcal{R}}^{MHO} \subseteq {\mathcal{R}}^{MHO}$.
\end{proof}

\subsection{Special Cases}
Consider the following special cases. We continue to use the terminology from the source coding literature.
\begin{enumerate}
\item \emph{Gaussian CEO hypothesis testing against independence:} When $R=0$, the problem reduces to the Gaussian CEO hypothesis testing against independence problem. Let $\mathcal{R}^{CEO}$ be the rate-exponent region of this problem. Define the set
\begin{align*}
\tilde{\mathcal{R}}^{CEO} \triangleq \Biggr \{&(R_1, \dots, R_L, E) : \hspace{0.05in}\textrm{there exists} \hspace{0.05in} \mb{r} \in \mathbb{R}_{+}^L \hspace{0.05in} \textrm{such that} \\
&\sum_{l \in S} R_l \ge \frac{1}{2} \log^{+} \left [ \frac{1}{D } \left (\frac{1}{\sigma^2_X} + \sum_{l \in  S^c} \frac{1- 2^{- 2 r_l}}{\sigma^2_{N_l}} \right )^{-1}\right ] + \sum_{l \in S} r_l \hspace{0.05in} \textrm{for all} \hspace{0.05in} S \subseteq \mathcal{L} \Biggr\}.
\end{align*}
We immediately have the following corollary as a consequence of Theorem 7.
\begin{Cor}
${\mathcal{R}}^{CEO} = \tilde{\mathcal{R}}^{CEO}.$
\end{Cor}
\item \emph{Gaussian one-helper hypothesis testing against independence:}
When $L=1$, the problem reduces to the Gaussian one-helper hypothesis testing against independence problem. Let $\mathcal{R}^{OH}$ be the rate-exponent region of this problem. Define the sets
\begin{align*}
\tilde{\mathcal{R}}^{OH} \triangleq \Biggr \{&(R, R_1, E) : \hspace{0.05in}\textrm{there exists} \hspace{0.05in} r_1 \in \mathbb{R}_{+} \hspace{0.05in} \textrm{such that} \\
&R_1 \ge r_1, \\
&R+R_1 \ge \frac{1}{2} \log^{+} \left [ \frac{\sigma^2_X}{D} \right ] + r_1, \hspace{0.05in} \textrm{and}\\
&R \ge \frac{1}{2} \log^{+} \left [ \frac{1}{D}  \left (\frac{1}{\sigma^2_X} + \frac{1- 2^{- 2 r_1}}{\sigma^2_{N_1}} \right )^{-1} \right ] \Biggr\},
\end{align*}
and
\begin{align*}
\bar{\mathcal{R}}^{OH} \triangleq  \Biggr\{&(R, R_1, E) : R \ge \frac{1}{2} \log^{+} \left [ \frac{\sigma^2_X}{ D} \left ( 1- \rho^2 +\rho^2 2^{-2 R_1}\right)\right ] \Biggr\},
\end{align*}
where
\[
\rho^2 = \frac{\sigma^2_X}{\sigma^2_X+\sigma^2_{N_1}}.
\]
\begin{Cor}
${\mathcal{R}}^{OH} = \tilde{\mathcal{R}}^{OH} = \bar{\mathcal{R}}^{OH}.$
\end{Cor}
\begin{proof}
The first equality follows from Theorem 7. Consider any $(R, R_1, E)$ in $\tilde{\mathcal{R}}^{OH}$. It must satisfy
\begin{align*}
R &\ge \min_{0 \le r_1 \le R_1} \max \left \{\frac{1}{2} \log^{+} \left [ \frac{1}{D } \left (\frac{1}{\sigma^2_X} + \frac{1- 2^{- 2 r_1}}{\sigma^2_{N_1}} \right )^{-1} \right ],\hspace{0.05in} \frac{1}{2} \log^{+} \left [ \frac{\sigma^2_X}{D} \right ] +  r_1 - R_1 \right \}\nonumber\\
&= \frac{1}{2} \log^{+} \left [ \frac{\sigma^2_X}{ D } \left ( 1- \rho^2 +\rho^2 2^{-2 R_1}\right)\right ],
\end{align*}
where the equality is achieved by
\begin{align}
r_1 = R_1 + \frac{1}{2} \log\left( 1- \rho^2 +\rho^2 2^{-2 R_1}\right).
\end{align}
We therefore have that $(R, R_1, E)$ is in $\bar{\mathcal{R}}^{OH}$, and hence $\tilde{\mathcal{R}}^{OH} \subseteq  \bar{\mathcal{R}}^{OH}$. The proof of the reverse containment follows by noticing that for any $(R, R_1, E)$ in $\bar{\mathcal{R}}^{OH}$, there exists $r_1$ as in (28) such that all inequalities in the definition of  $\tilde{\mathcal{R}}^{OH}$ are satisfied.
\end{proof}
\end{enumerate}

\section{A General Outer Bound}
We return to the general problem formulated in Section 3. The problem remains open till date. Several inner bounds are known for $L=1$ \cite{Han, Ahl, Han1, Han2}. But even for $L=1$, there is no nontrivial outer bound with which to compare the inner bounds. We give an outer bound for a class of instances of the general problem.

Consider the class of instances such that $P_\mb{X} = Q_\mb{X}$, i.e., the marginal distributions of $\mb{X}$ are the same under both hypotheses. Stein's lemma \cite{Cover} asserts that the centralized type 2 error exponent for this class of problems is
\[
E_C \triangleq D\left(P_{\mb{X}Y} \| Q_{\mb{X}Y}\right),
\]
which is achieved when $\mb{X}$ and $Y$ both are available at the detector. Let
\[
\mathcal{R}_{C} \triangleq \{(\mb{R},E) : E \le E_C\}.
\]
We have the following trivial centralized outer bound.
\begin{Lem}
$\mathcal{R} \subseteq \mathcal{R}_{C}.$
\end{Lem}

Let $\Xi$ be the set of random variables $Z$ such that there exists two joint distributions $P_{\mb{X}YZ}$ and $Q_{\mb{X}YZ}$ satisfying
\begin{enumerate}
\item[(C12)] $\sum_{\mathcal{Z}} P_{\mb{X}YZ}=P_{\mb{X}Y},$ the distribution under $H_0$,
\item[(C13)] $\sum_{\mathcal{Z}} Q_{\mb{X}YZ}=Q_{\mb{X}Y},$ the distribution under $H_1$,
\item[(C14)] $Q_{\mb{X}YZ} = Q_{\mb{X}|Z} Q_{Y|Z}Q_Z,$ i.e., $\mb{X}$ and $Y$ are conditionally independent given $Z$ under the $Q$ distribution, and
\item[(C15)] $P_{\mb{X}Z} = Q_{\mb{X}Z},$ i.e., the joint distributions of $(\mb{X},Z)$ are the same under both distributions.
\end{enumerate}
Note that the joint distributions of $(Y,Z)$ need not be the same under the two distributions. If $P_{\mb{X}YZ}$ and $Q_{\mb{X}YZ}$ are the joint distributions of $\mb{X}$, $Y$, and $Z$ under $H_0$ and $H_1$, respectively and $Z$ is available to the detector, then the problem can be related to the $L$-encoder hypothesis testing against conditional independence. Now $Z$ is not present in the original problem, but we can augment the sample space by introducing $Z$ and supplying it to the decoder. The outer bound for this new problem is then an outer bound for the original problem. Moreover, we can then optimize over $Z$ to obtain the best possible bound.

Let $\chi$ and $\Lambda_o$ be defined as in Section 4.2 with $X_{L+1}$ restricted to be deterministic. If $\Xi$ is nonempty, then for any $(Z,X,\lambda_o)$ in $\Xi \times \chi \times \lambda_o$, define the set
\begin{align*}
\mathcal{R}_o(Z,X,\lambda_o) \triangleq \biggr \{ (\mb{R},E) :
\sum_{l \in S}R_l &\ge I({X}; \mb{U}_S|\mb{U}_{S^c},Z,T) + \sum_{l \in S} I(X_l ; U_l | X, W, Z,T) \hspace{0.05in} \textrm{for all} \hspace{0.05in}  S \subseteq \mathcal{L}, \hspace {0.05 in} \textrm{and} \\
E &\le I(Y;\mb{U}|Z,T)+D\left(P_{Y|Z}\|Q_{Y|Z}|Z\right) \biggr \}.\nonumber
\end{align*}
Finally, let
\[
\mathcal{R}_o \triangleq \left\{
\begin{array}{l l}
  \bigcap_{Z \in \Xi} \bigcap_{X \in \chi} \bigcup_{\lambda_o \in \Lambda_o} \mathcal{R}_o(Z,X,\lambda_o) & \quad \mbox{if $\Xi$ is nonempty}\\
  \mathbb{R}^{L+1}_{+} & \quad \mbox{otherwise.}\\ \end{array} \right.
\]
We have the following outer bound to the rate-exponent region of this class of problems.
\begin{Thm}
$\mathcal{R} \subseteq \overline{\mathcal{R}_o} \cap \mathcal{R}_{C}.$
\end{Thm}
\begin{proof}
In light of Proposition 1 and Lemma 5, it suffices to show that
\[
\mathcal{R}_{*} \subseteq \mathcal{R}_o.
\]
Consider $(\mb{R},E)$ in $\mathcal{R}_{*}$.Then there exists a block length $n$ and encoders $f_l^{(n)}$ such that
\begin{align}
R_l &\ge \frac{1}{n}\log \left|f_l^{(n)}\left({X}_l^n\right)\right|\hspace{0.05in} \textrm{for all \emph{l} in} \hspace{0.05in}  \mathcal{L}, \hspace {0.05 in} \textrm{and} \\
E &\le \frac{1}{n} D\biggr(P_{\left (f_l^{(n)} \left ({X}_l^n \right)\right)_{l \in \mathcal{L}}{Y}^n} \Bigr\| Q_{\left (f_l^{(n)} \left ({X}_l^n \right)\right)_{l \in \mathcal{L}}{Y}^n}\biggr).
\end{align}
Consider any $Z$ in $\Xi$. Then
\begin{align*}
D&\biggr(P_{\left (f_l^{(n)} \left ({X}_l^n \right)\right)_{l \in \mathcal{L}}{Y}^n} \Bigr\| Q_{\left (f_l^{(n)} \left ({X}_l^n \right)\right)_{l \in \mathcal{L}}{Y}^n}\biggr) \\
&\le D\biggr(P_{\left (f_l^{(n)} \left ({X}_l^n \right)\right)_{l \in \mathcal{L}}{Y}^nZ^n} \Bigr\| Q_{\left (f_l^{(n)} \left ({X}_l^n \right)\right)_{l \in \mathcal{L}}{Y}^nZ^n}\biggr) \nonumber\\
&= D\biggr(P_{\left (f_l^{(n)} \left ({X}_l^n \right)\right)_{l \in \mathcal{L}}{Y}^n\bigr|Z^n} \Bigr\| Q_{\left (f_l^{(n)} \left ({X}_l^n \right)\right)_{l \in \mathcal{L}}{Y}^n\bigr|Z^n} \Bigr| Z^n\biggr) \nonumber\\
&= D\biggr(P_{\left (f_l^{(n)} \left ({X}_l^n \right)\right)_{l \in \mathcal{L}}{Y}^n\bigr|Z^n} \Bigr\| P_{\left (f_l^{(n)} \left ({X}_l^n \right)\right)_{l \in \mathcal{L}}\bigr|Z^n}Q_{{Y}^n|Z^n} \Bigr| Z^n\biggr) \nonumber\\
&= D\biggr(P_{\left (f_l^{(n)} \left ({X}_l^n \right)\right)_{l \in \mathcal{L}}{Y}^n\bigr|Z^n} \Bigr\| P_{\left (f_l^{(n)} \left ({X}_l^n \right)\right)_{l \in \mathcal{L}}\bigr|Z^n}P_{{Y}^n|Z^n} \Bigr| Z^n\biggr) + D\Bigr(P_{{Y}^n|Z^n} \bigr\| Q_{{Y}^n|Z^n} \bigr| Z^n\Bigr) \nonumber\\
&= I\biggr(\left (f_l^{(n)} \left ({X}_l^n \right)\right)_{l \in \mathcal{L}};{Y}^n \Bigr|Z^n\biggr)+ nD\left(P_{{Y}|Z} \| Q_{{Y}|Z} | Z\right), \nonumber
\end{align*}
which together with (30) implies
\begin{align}
E \le  \frac{1}{n}I\biggr(\left (f_l^{(n)} \left ({X}_l^n \right)\right)_{l \in \mathcal{L}};{Y}^n \Bigr|Z^n\biggr) + D\left(P_{{Y}|Z} \| Q_{{Y}|Z} | Z\right).
\end{align}
It now follows from (29), (31), and Corollary 1 that $\left(\mb{R},\left(E-D\left(P_{{Y}|Z} \| Q_{{Y}|Z} | Z\right)\right)^{+}\right)$ is in $\mathcal{R}_{*}^{CI}$. Therefore from Theorem 2, it must also be in $\mathcal{R}^{CI}_o$. Hence for any $X$ in $\chi$, there exists $\lambda_o$ in $\Lambda_o$ such that $\left(\mb{R},\left(E-D\left(P_{{Y}|Z} \| Q_{{Y}|Z} | Z\right)\right)^{+}\right)$ is in $\mathcal{R}_o^{CI}(X,\lambda_o)$, i.e.,
\begin{align*}
\sum_{l \in S}R_l &\ge I({X}; \mb{U}_S|\mb{U}_{S^c},Z,T) + \sum_{l \in S} I(X_l ; U_l | X, W, Z,T)\hspace{0.05in} \textrm{for all} \hspace{0.05in}  S \subseteq \mathcal{L}, \hspace {0.05 in} \textrm{and} \\
\left(E-D\left(P_{{Y}|Z} \| Q_{{Y}|Z} | Z\right)\right)^{+} &\le I(Y;\mb{U}|Z,T).\nonumber
\end{align*}
This means that $(\mb{R},E)$ is in $\mathcal{R}_o(Z,X,\lambda_o)$, and hence $\mathcal{R}_{*} \subseteq \mathcal{R}_o$.
\end{proof}
Although the outer bound above is not computable in general, it simplifies to the following computable form for the special case in which $L=1$. Let
\begin{align*}
\tilde {\mathcal{R}} &\triangleq \bigcap_{Z \in \Xi}\Bigr \{(R_1,E): \hspace{0.05in}\textrm{there exists} \hspace{0.05in} U_1 \hspace{0.05in}\textrm{such that} \nonumber\\
&\hspace{1in}R_1 \ge I(X_1 ;  U_1 | Z), \nonumber\\
&\hspace{1in}E \le I(Y;{U_1}|Z)+D\left(P_{{Y}|Z} \| Q_{{Y}|Z} | Z\right), \\
&\hspace{1in}| {\mathcal{U}}_1| \le |\mathcal{X}_1| + 1, \hspace{0.05in} \textrm{and} \nonumber\\
&\hspace{1in} U_1 \leftrightarrow X_1 \leftrightarrow (Y, Z)\Bigr \}.\nonumber
\end{align*}
\begin{Cor}For $1$-encoder general hypothesis testing, $\overline{\mathcal{R}_o}=\tilde {\mathcal{R}}$ and hence $\mathcal{R} \subseteq \tilde {\mathcal{R}} \cap \mathcal{R}_{C}.$
\end{Cor}
\begin{proof}
It suffices to show that $\overline{\mathcal{R}_o}=\tilde {\mathcal{R}}$. This immediately follows by noticing that given any $Z$ in $\Xi$, the outer bound can be related to the rate-exponent region of the $1$-encoder hypothesis testing against conditional independence problem. The result then follows from Theorem 3.
\end{proof}
It is easy to see that the outer bound is tight for the test against independence.

\begin{Cor} (Test against independence, \cite{Ahl}) If $Q_{X_1Y} = P_{X_1} P_Y$, then
\[
\mathcal{R} = \tilde{\mathcal{R}}.
\]
\end{Cor}
\begin{proof}
This follows by choosing $Z$ to be deterministic in the outer bound and then
invoking the result of Ahlswede and Csisz\'{a}r~\cite{Ahl}.
\end{proof}

\emph{Remark 2:} The outer bound is not always better than the centralized outer bound. In particular, if
\[
D\left(P_{{Y}|Z} \| Q_{{Y}|Z} | Z\right) \ge E_C
\]
for all $Z$ in $\Xi$, then the outer bound is no better than the centralized outer bound.

\subsection{Gaussian Case}
To illustrate this bound, let us consider a Gaussian example in which $X_1$ and $Y$ are zero-mean unit-variance jointly Gaussian sources with the correlation coefficients $\rho_0$ and $\rho_1$ under $H_0$ and $H_1$, respectively, where $\rho_0 \neq \rho_1$, $\rho_0^2 < 1,$ and $\rho_1^2 < 1$. We can assume without loss of generality that $0 \le \rho_1 < 1$ because the case $-1 < \rho_1 \le 0$ can be handled by multiplying $Y$ by $-1$. We use lowercase $p$ and $q$ to denote appropriate Gaussian densities under hypotheses $H_0$ and $H_1$, respectively. Let $\mathcal{R}^{G}$ be the rate-exponent region of this problem. We focus on the following three regions (Fig. \ref{fig:Fig7}) for which the outer bound is nontrivial.
\begin{align*}
\mathcal{D}_1 &\triangleq \{(\rho_0,\rho_1) : 0 \le \rho_1 < \rho_0 < 1\}, \\
\mathcal{D}_2 &\triangleq \{(\rho_0,\rho_1) : 0 \le \rho_1 \hspace{0.05in}\textrm{and}\hspace{0.05in} 2 \rho_1 - 1 \le \rho_0 < \rho_1\}, \\
\mathcal{D}_3 &\triangleq \left \{(\rho_0,\rho_1) : -1 < \rho_0 \le 2 \rho_1 - 1  \hspace{0.05in}\textrm{and}\hspace{0.05in} \frac{ 2(\log e) \rho_1}{1-\rho_1} \le \frac{1}{2} \log \left(\frac{1-\rho_1^2}{1-\rho_0^2}\right) - \frac{(\log e) \rho_1(\rho_0-\rho_1)}{1-\rho_1^2} \right\}.
\end{align*}
\begin{figure}[htp]
\centering
\includegraphics[width=3.6in, totalheight=0.2\textheight,viewport=10 60 730 500,clip]{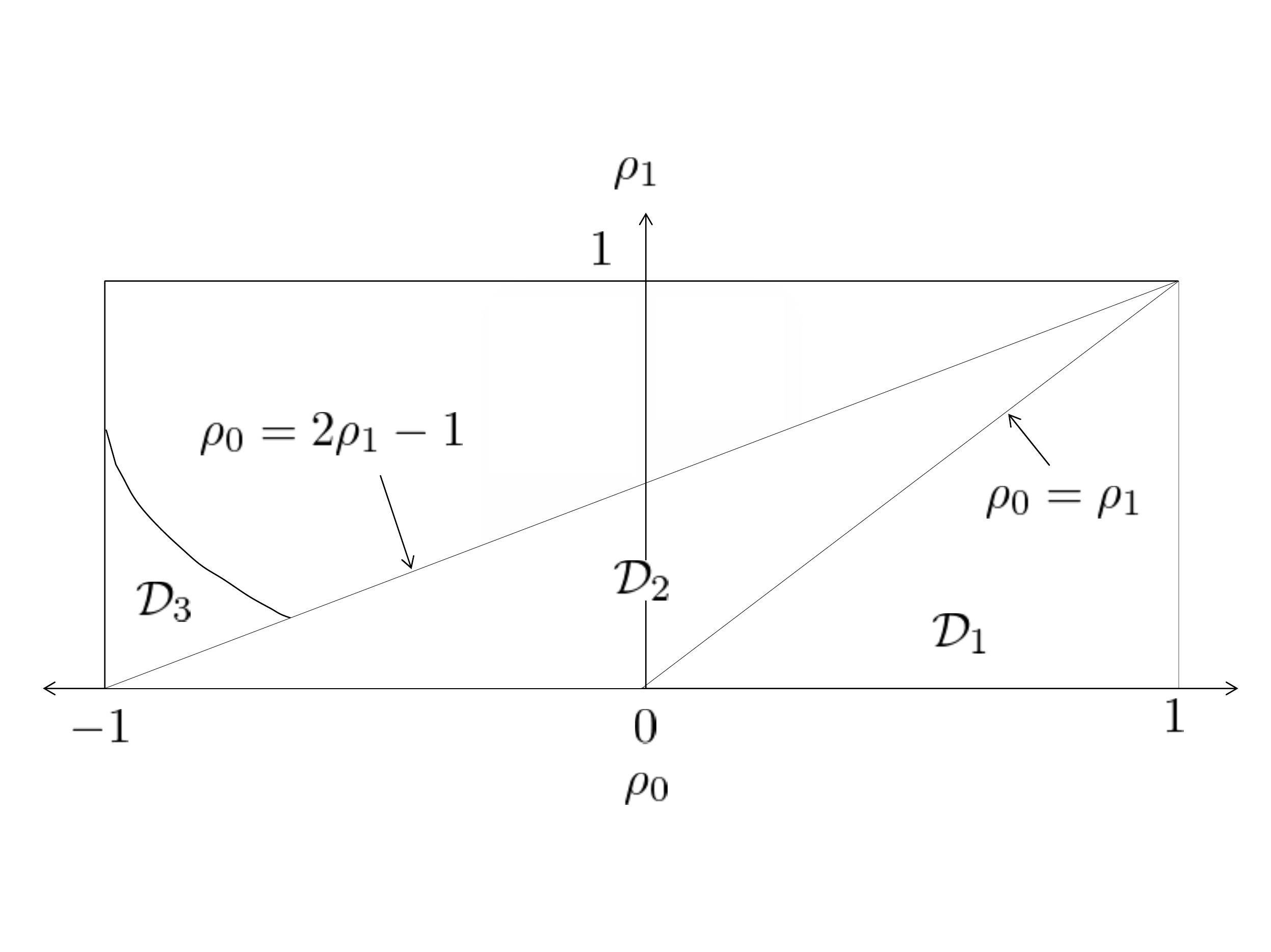}
\caption{Regions of pair $(\rho_0,\rho_1)$ for which the outer bound is nontrivial}\label{fig:Fig7}
\end{figure}
\subsubsection{Outer Bound}
Let us define
\begin{align*}
\rho &\triangleq \left\{
\begin{array}{l l}
  \frac{\rho_0 - \rho_1}{1 - \rho_1} & \quad \mbox{if $(\rho_0,\rho_1)$ is in $\mathcal{D}_1 \cup \mathcal{D}_2$}\\
  \frac{\rho_0 + \rho_1}{1 - \rho_1} & \quad \mbox{if $(\rho_0,\rho_1)$ is in $\mathcal{D}_3$}. \\ \end{array} \right.
\end{align*}
and
\begin{align*}
C &\triangleq \left\{
\begin{array}{l l}
  0 & \quad \mbox{if $(\rho_0,\rho_1)$ is in $\mathcal{D}_1 \cup \mathcal{D}_2$}\\
  \frac{2 (\log e) \rho_1}{1 - \rho_1} & \quad \mbox{if $(\rho_0,\rho_1)$ is in $\mathcal{D}_3$}. \\ \end{array} \right.
\end{align*}
The centralized type 2 error exponent is
\begin{align*}
E^G_C &\triangleq D\left(p_{{X}_1Y} \| q_{{X}_1Y}\right) \\
&= \frac{1}{2} \log \left(\frac{1-\rho_1^2}{1-\rho_0^2}\right) - \frac{(\log e) \rho_1(\rho_0-\rho_1)}{1-\rho_1^2}.
\end{align*}
Define the sets
\begin{align*}
{\mathcal{R}}^{G}_o \triangleq & \Biggr \{(R_1,E) : E \le \frac{1}{2} \log \left ( \frac{1}{1 - {\rho}^2 + {\rho}^2 2^{-2R_1}}\right) + C\Biggr\}
\end{align*}
and
\[
\mathcal{R}^G_{C} \triangleq \left\{({R_1},E) : E \le E^G_C\right\}.
\]
We have the following outer bound.
\begin{Thm}
If $(\rho_0,\rho_1)$ is in $\mathcal{D}_1 \cup \mathcal{D}_2 \cup \mathcal{D}_3$, then
$$\mathcal{R}^{G} \subseteq {\mathcal{R}}^{G}_o \cap \mathcal{R}^{G}_C.$$
\end{Thm}
\begin{proof}
The proof is in two steps: obtain a single letter outer bound similar to the one in Corollary 5 and then use it to obtain the desired outer bound. Consider $(\rho_0,\rho_1)$ in $\mathcal{D}_1$. Let $Z$, $Z^{'}$, $W$, and $V$ be standard normal random variables independent of each other. ${X}_1$ and ${Y}$ can be expressed as
\begin{align*}
X_1 &= \sqrt{\rho_1} Z + \sqrt{\rho_0-\rho_1} Z^{'} + \sqrt{1-\rho_0} W \\
Y &= \sqrt{\rho_1} Z + \sqrt{\rho_0-\rho_1} Z^{'}+ \sqrt{1-\rho_0} V
\end{align*}
under $H_0$ and as
\begin{align*}
X_1 &= \sqrt{\rho_1} Z + \sqrt{1-\rho_1} W \\
Y &= \sqrt{\rho_1} Z + \sqrt{1-\rho_1} V
\end{align*}
under $H_1$. It is easy to verify that conditions (C12) through (C15) are satisfied if we replace the distributions by the corresponding Gaussian densities. Therefore, $Z$ is in $\Xi$. Define the set
\begin{align*}
\tilde {\mathcal{R}}^{G} \triangleq \Bigr \{(R_1,E) : \hspace{0.05in}&\textrm{there exists} \hspace{0.05in} U_1 \hspace{0.05in}\textrm{such that} \nonumber\\
&R_1 \ge I({X}_1; U_1|{Z}), \\
&E \le I({Y} ;U_1 |{Z})+D(p_{Y|Z}\|q_{Y|Z}|Z),\hspace{0.05in} \textrm{and}\\
&({Y,Z}) \leftrightarrow {X}_1\leftrightarrow U_1 \Bigr \}.
\end{align*}
\begin{Cor}
$\mathcal{R}^{G} \subseteq \overline{\tilde {\mathcal{R}}^{G}} \cap \mathcal{R}^{G}_C$.
\end{Cor}

The proof is immediate as a continuous extension of Corollary 5. From Corollary 7, it suffices to show that
\[
\tilde {\mathcal{R}}^{G} \subseteq {\mathcal{R}}^{G}_o.
\]
Note first that
\begin{align*}
D(p_{Y|Z}\|q_{Y|Z}|Z) = 0
\end{align*}
here because the joint densities of $(Y,Z)$ are the same under both hypotheses. Consider any $(R_1,E)$ in $\tilde {\mathcal{R}}^{G}$. Then there exists a random variable $U_1$ such that $({Y,Z}) \leftrightarrow {X}_1\leftrightarrow U_1$,
\begin{align}
R_1 &\ge I({X}_1; U_1|{Z}), \hspace{0.05in} \textrm{and}\\
E &\le I({Y} ;U_1 |{Z}).
\end{align}
Since $X_1, Y$, and $Z$ are jointly Gaussian under $H_0$, we can write that
\[
Y = {\rho} X_1 + \sqrt{\rho_1} (1 - {\rho}) Z + B,
\]
where $B$ is a zero-mean Gaussian random variable with the variance
\[
\sigma^2_{Y|X_1Z} = (1 - \rho_1)\left (1-{\rho}^2\right),
\]
and is independent of $X_1$ and $Z$. We now have
\begin{align}
h(Y|U_1,Z) &= h\left({\rho} X_1 + \sqrt{\rho_1 } (1 - {\rho}) Z + B|U_1,Z\right) \nonumber\\
&= h\left({\rho} X_1 + B|U_1,Z\right) \nonumber\\
&\ge \frac{1}{2} \log \left( 2^{2h({\rho}X_1|U_1,Z)} + 2^{2h(B)} \right) \\
&= \frac{1}{2} \log \left( {\rho}^2 2^{2h(X_1|U_1,Z)} + 2^{2h(B)} \right) \nonumber\\
&= \frac{1}{2} \log \left( {\rho}^2 2^{2\left(h(X_1|Z) - I(X_1;U_1|Z)\right)} + 2^{2h(B)} \right) \nonumber\\
&= \frac{1}{2} \log \left( {\rho}^2 (1 - \rho_1) 2^{-2 I(X_1;U_1|Z)} + (1 - \rho_1)\left (1-{\rho}^2 \right) \right) + \frac{1}{2} \log (2 \pi e)\nonumber\\
&\ge \frac{1}{2} \log \left( {\rho}^2 (1 - \rho_1) 2^{-2 R_1} + (1 - \rho_1)\left (1-{\rho}^2 \right)\right) + \frac{1}{2} \log (2 \pi e),
\end{align}
where
\begin{enumerate}
\item[(34)] follows from the entropy power inequality \cite{Cover} because $X_1$ and $B$ are independent given $(U_1,Z)$, and
\item[(35)] follows because function
\[
f(x) = \frac{1}{2} \log \left( p 2^{-2x} + q \right)
\]
is monotonically decreasing in $x$ for $p > 0$, and we have the rate constraint in (32).
\end{enumerate}
Now (33) and (35) imply
\begin{align*}
E &\le \frac{1}{2}\log \left (\frac{\sigma^2_{Y|Z}}{{\rho}^2 (1 - \rho_1) 2^{-2 R_1} + (1 - \rho_1)\left (1-{\rho}^2 \right)}  \right ) \nonumber\\
&=\frac{1}{2} \log \left ( \frac{1}{1 - {\rho}^2 + {\rho}^2 2^{-2R_1}}\right),
\end{align*}
which proves that $(R_1,E)$ is in ${\mathcal{R}}^{G}_o$. This completes the proof for the region $\mathcal{D}_1$.

The proof is analogous for $(\rho_0,\rho_1)$ in the region $\mathcal{D}_2$. The only difference is that under $H_0$, ${X}_1$ and ${Y}$ can now be expressed as
\begin{align*}
X_1 &= \sqrt{\rho_1} Z + \sqrt{\rho_1-\rho_0} Z^{'} + \sqrt{1-2 \rho_1+\rho_0} W \\
Y &= \sqrt{\rho_1} Z - \sqrt{\rho_1-\rho_0} Z^{'}+ \sqrt{1-2 \rho_1+\rho_0} V.
\end{align*}

Suppose now that $(\rho_0,\rho_1)$ is in $\mathcal{D}_3$. One can verify that $-\rho_0-\rho_1 > 0$ here. Hence, $X_1$ and $Y$ can be expressed as
\begin{align*}
X_1 &= \sqrt{\rho_1} Z + \sqrt{-\rho_0-\rho_1} Z^{'} + \sqrt{1+\rho_0} W \\
Y &= -\sqrt{\rho_1} Z - \sqrt{-\rho_0-\rho_1} Z^{'}+ \sqrt{1+\rho_0} V
\end{align*}
under $H_0$. Their expressions under $H_1$ are the same as before. It is evident that $Z$ is in $\Xi$. Therefore, the outer bound in Corollary 7 is valid for this case, which implies that it suffices to show that
\[
\tilde{\mathcal{R}}^{G} \subseteq {\mathcal{R}}^{G}_o.
\]
Under $H_0$, the conditional distribution of $Y$ given $Z=z$ is Gaussian with the mean $-\sqrt{\rho_1}z$ and the variance $1-\rho_1$. Similarly under $H_1$, it is Gaussian with the mean $\sqrt{\rho_1}z$ and the variance $1-\rho_1$. We therefore obtain
\begin{align*}
D(p_{Y|Z}\|q_{Y|Z}|Z) &= \int_{z\in \mathbb{R}} p_Z(z) dz\int_{y\in \mathbb{R}} p_{Y|Z}(y|z) \log \frac{p_{Y|Z}(y|z)}{q_{Y|Z}(y|z)} dy \\
&= \int_{z\in \mathbb{R}} p_Z(z) dz\int_{y\in \mathbb{R}} p_{Y|Z}(y|z) \log \left[\exp\left(\frac{(y-\sqrt{\rho_1}z)^2}{2(1-\rho_1)}-\frac{(y+\sqrt{\rho_1}z)^2}{2(1-\rho_1)}\right) \right] dy\\
&= \int_{z\in \mathbb{R}} p_Z(z) dz\int_{y\in \mathbb{R}} p_{Y|Z}(y|z) \left[-\frac{2(\log e)\sqrt{\rho_1}yz}{1-\rho_1}\right] dy\\
&= -\frac{2(\log e)\sqrt{\rho_1}}{1-\rho_1} \int_{z\in \mathbb{R}} z p_Z(z) dz\int_{y\in \mathbb{R}} yp_{Y|Z}(y|z)  dy\\
&= -\frac{2(\log e)\sqrt{\rho_1}}{1-\rho_1} \int_{z\in \mathbb{R}} z p_Z(z) dz\left(-\sqrt{\rho_1}z\right)\\
&= \frac{2(\log e)\rho_1}{1-\rho_1} \int_{z\in \mathbb{R}} z^2 p_Z dz\\
&=\frac{2 (\log e) \rho_1}{1 - \rho_1}.
\end{align*}
Again, since $X_1, Y,$ and $Z$ are jointly Gaussian under $H_0$, we can write
\[
Y = {\rho} X_1 - \sqrt{\rho_1} (1 + {\rho}) Z + B,
\]
where $B$ is defined as before. The rest of the proof is identical to the region $\mathcal{D}_1$ case.
\end{proof}

\subsubsection{Ahlswede and Csisz\'{a}r's Inner Bound}
We next compare the outer bound with Ahlswede and Csisz\'{a}r's inner bound, which is obtained by using a Gaussian test channel to quantize $X_1$. One can use better inner bounds \cite{Han1, Han2}, but they are quite complicated and for the Gaussian case considered here, Ahlswede and Csisz\'{a}r's bound itself is quite close to our outer bound in some cases. Let
\begin{align*}
{\mathcal{R}}^{G}_i \triangleq & \Biggr \{(R_1,E) : E \le  \frac{1}{2} \log \lsb \frac{1 - \rho^2_1 \lsb 1 - 2^{-2 R_1} \rsb }{1 - \rho^2_0 \lsb 1 - 2^{-2 R_1} \rsb }\rsb - \frac{(\log e)\rho_1 \lsb \rho_0 - \rho_1 \rsb  \lsb 1 - 2^{-2 R_1} \rsb}{1 - \rho^2_1 \lsb 1 - 2^{-2 R_1} \rsb } \Biggr\}.
\end{align*}
\begin{Prop} \cite{Ahl}
${\mathcal{R}}^{G}_i \subseteq {\mathcal{R}}^{G}.$
\end{Prop}
\begin{proof}
Fix any $(R_1,E)$ in ${\mathcal{R}}^{G}_i$. Let $U_1 = X_1 + P$, where $P$ is a zero-mean Gaussian random variable independent of $(X_1,Y)$ such that
\[
I(X_1;U_1) = R_1,
\]
which implies that the variance of $P$
\[
\sigma^2_P = \frac{1}{2^{2R_1}-1}.
\]
The covariance matrix of $({U}_1, {Y})$ is
\[
\mb{K}_0 = {\left [ \begin{array}{cc} 1+\sigma^2_P & \rho_0 \\
\rho_0 & 1 \end{array} \right]}
\]
under $H_0$ and is
\[
\mb{K}_1 = {\left [ \begin{array}{cc} 1+\sigma^2_P & \rho_1 \\
\rho_1 & 1 \end{array} \right]}.
\]
under $H_1$. It now follows from Ahlswede and Csisz\'{a}r's scheme \cite[Theorem 5]{Ahl} that the achievable exponent is
\begin{align*}
E_{AC} &= D(p_{U_1Y}\| q_{U_1Y}) \\
&= \int_{\mb{z} \in \mathbb{R}^2} p_{U_1Y} (\mb{z}) \log \frac{p_{U_1Y}(\mb{z})}{q_{U_1Y}(\mb{z})} d \mb{z}\\
&= -\frac{1}{2} \log \left((2 \pi e)^2 \det(\mb{K}_0)\right) - \int_{\mb{z} \in \mathbb{R}^2} p_{U_1Y} (\mb{z}) \log {q_{U_1Y} (\mb{z})} d \mb{z}\\
&= -\frac{1}{2} \log \left((2 \pi e)^2 \det(\mb{K}_0)\right) - \int_{\mb{z} \in \mathbb{R}^2} p_{U_1Y} (\mb{z}) \left [-\frac{(\log e)}{2} \mb{z}^T \mb{K}_1^{-1} \mb{z} -\frac{1}{2} \log \left((2 \pi)^2 \det(\mb{K}_1)\right) \right ] d \mb{z}\\
&= \frac{1}{2} \log \frac{\det(\mb{K}_1)}{\det(\mb{K}_0)} - (\log e) + \frac{(\log e)}{2} \int_{\mb{z} \in \mathbb{R}^2} p_{U_1Y} (\mb{z}) \left(\mb{z}^T \mb{K}_1^{-1} \mb{z}\right) d \mb{z}\\
&= \frac{1}{2} \log \frac{\det(\mb{K}_1)}{\det(\mb{K}_0)} - (\log e) + \frac{(\log e)(1+\sigma^2_P-\rho_0 \rho_1)}{\det(\mb{K}_1)} \\
&= \frac{1}{2} \log \frac{(1+\sigma^2_P- \rho_1^2)}{(1+\sigma^2_P- \rho_0^2)} - \log e + \frac{(\log e)(1+\sigma^2_P-\rho_0 \rho_1)}{(1+\sigma^2_P- \rho_1^2)} \\
&=\frac{1}{2} \log \lsb \frac{1 - \rho^2_1 \lsb 1 - 2^{-2 R_1} \rsb }{1 - \rho^2_0 \lsb 1 - 2^{-2 R_1} \rsb }\rsb - \frac{(\log e)\rho_1 \lsb \rho_0 - \rho_1 \rsb  \lsb 1 - 2^{-2 R_1}\rsb}{1 - \rho^2_1 \lsb 1 - 2^{-2 R_1}\rsb }.
\end{align*}
This proves that $(R_1,E)$ is in $\mathcal{R}^G$.
\end{proof}
The inner and outer bounds coincide for the test against independence.
\begin{Cor} (Test against independence, \cite{Ahl,Oohama}) If ${X}_1$ and ${Y}$ are independent under $H_1$, i.e., $\rho_1 = 0$, then
\[
\mathcal{R}^{G} = {\mathcal{R}}^{G}_o =  {\mathcal{R}}^{G}_i.
\]
\end{Cor}

\subsubsection{Numerical Results}
\begin{figure}[htp]
\centering
  \includegraphics[width=6.5in]{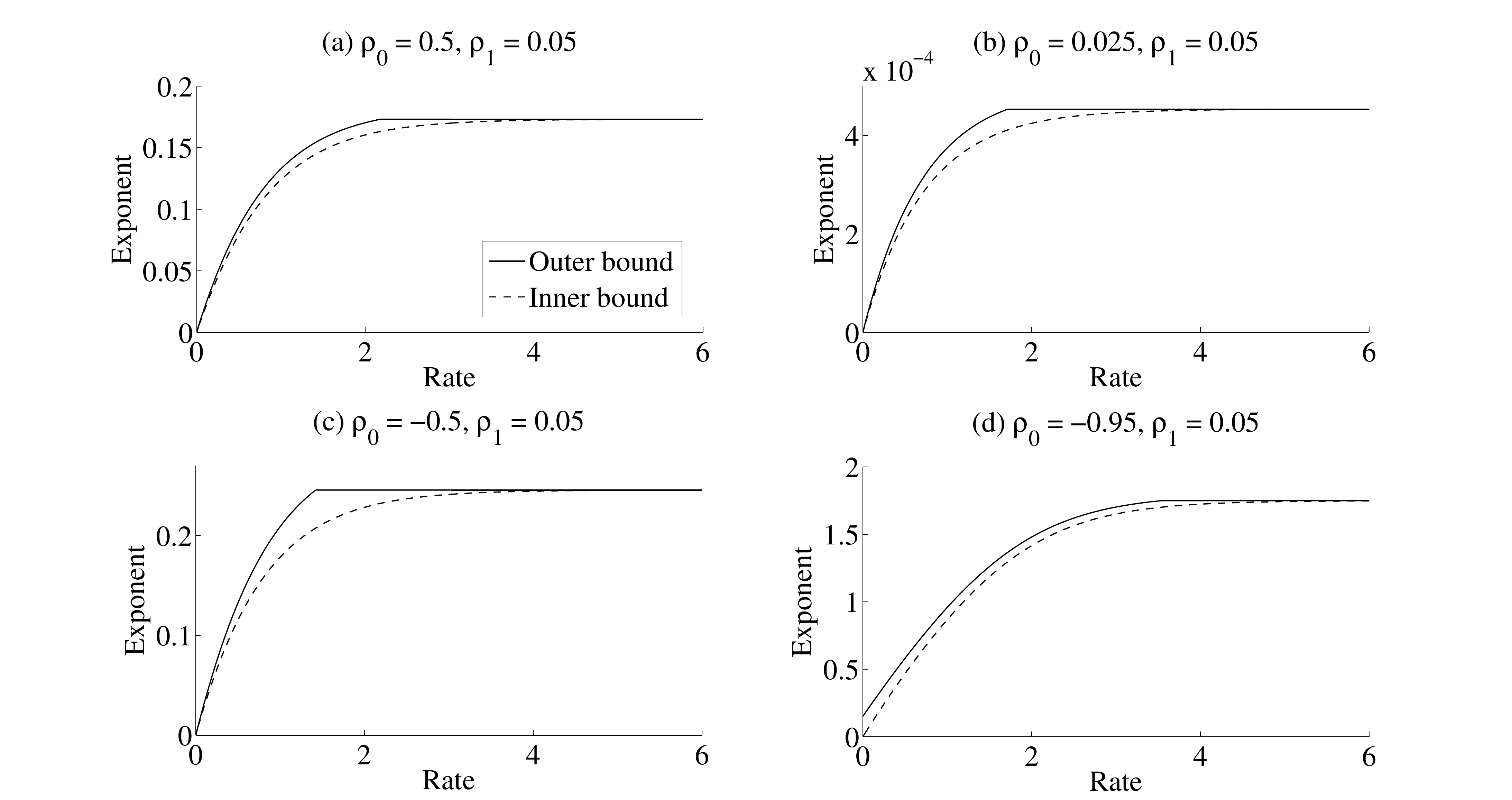}
\caption{Outer and inner bounds for four examples}\label{fig:Fig8}
\end{figure}
Fig. \ref{fig:Fig8} shows the inner and outer bounds for four examples. Fig. \ref{fig:Fig8}(a)-(c) are the examples when $(\rho_0,\rho_1)$ is in $\mathcal{D}_1 \cup \mathcal{D}_2$. Observe that the two bounds are quite close near zero and at all large rates. Fig. \ref{fig:Fig8}(d) is an example when $(\rho_0,\rho_1)$ is in $\mathcal{D}_3$. For this example, there is a gap between the inner and outer bounds at zero rate. This is due to the fact that in our outer bound, the joint densities of $(Y,Z)$ are different under the two hypotheses. Numerical results suggest that for a fixed $\rho_0$, the maximum gap between the inner and outer bounds decreases as we decrease $\rho_1$ and finally becomes zero at $\rho_1 = 0$, which is the test against independence.

\emph{Remark 3:} The outer bound can be extended to the vector Gaussian case. One can obtain a single letter outer bound similar to the one in Corollary 7. Then the outer bound can be optimized over all choices of $U_1$ by using an invertible transformation \cite{Chao, Globerson} and the scalar solution obtained above. It follows from our earlier work that the outer bound is tight for the test against independence \cite{Rahman}.

\section*{Acknowledgment}
This research was supported by the Air Force Office of Scientific Research (AFOSR) under grant FA9550-08-1-0060.\\ \\
\Large
{\textbf{Appendix A:}\hspace{0.05in} \textbf{Proof of Lemma 1}}\newline
\normalsize \\
The proof is rather well known and appears in source coding literature quite often. For instance, the similar proof can be found in \cite{Wagner3}. Let us define
\begin{align*}
\bar {\Lambda}_i \triangleq \Bigr \{ \lambda_i = (\mb{U},T) \in {\Lambda}_i : |\mathcal{U}_l| &\le |\mathcal{X}_l| + 2^{L} - 1 \hspace{0.05in}\textrm{for all} \hspace{0.05in} l \in \mathcal{L}, \hspace{0.05in}\textrm{and} \\
|\mathcal{T}| &\le 2^{L} \Bigr \},
\end{align*}
and
\begin{align*}
\bar {\mathcal{R}}^{CI}_i \triangleq \bigcup_{\lambda_i \in \bar {\Lambda}_i} \mathcal{R}_i^{CI}(\lambda_i).\nonumber
\end{align*}
We want to show that ${\mathcal{R}}_i^{CI}=\bar {\mathcal{R}}_i^{CI}$. We start with the deterministic $T$ case. Consider $\lambda_i = (\mb{U},T)$ in $\Lambda_i$, where $T$ is deterministic. For any $S \subseteq \mathcal{L}$ containing 1, we have
\[
I(\mb{X}_S;\mb{U}_S | \mb{U}_{S^c},X_{L+1},Z) = H(\mb{X}_S | \mb{U}_{S^c},X_{L+1},Z) - H(\mb{X}_S | \mb{U}_{1^c}, U_1, X_{L+1},Z),
\]
and for any nonempty $S$ not containing 1, we have
\[
I(\mb{X}_S;\mb{U}_S | \mb{U}_{S^c},X_{L+1},Z) = I(\mb{X}_S;\mb{U}_S | \mb{U}_{S^c \setminus \{1\}}, U_1 , X_{L+1},Z).
\]
Moreover,
\[
I(Y; \mb{U},X_{L+1}|Z) = H(Y|X_{L+1},Z) - H(Y | \mb{U}_{1^c}, U_1,X_{L+1},Z).
\]
It follows from the support lemma \cite[Lemma 3.4, pp. 310]{csiszar} that there exists $\bar {U}_1$ with $\bar {\mathcal{U}}_1 \subseteq {\mathcal{U}}_1$ such that
\[|\bar {\mathcal{U}}_1| \le |\mathcal{X}_1| + 2^{L} - 1,
\]
\[
\sum_{u_1 \in \bar {\mathcal{U}}_1} \textrm{Pr}(X_1=x_1| U_1 = u_1) \textrm{Pr}(\bar U_1 = u_1) = \textrm{Pr}(X_1=x_1) \hspace{0.05in} \textrm{for all} \hspace{0.05in} x_1 \hspace{0.05in}\textrm{in} \hspace{0.05in}\mathcal{X}_1\hspace{0.05in} \textrm{but one},
\]
\begin{align*}
 H(\mb{X}_S | \mb{U}_{1^c}, U_1, X_{L+1},Z) =  H(\mb{X}_S | \mb{U}_{1^c}, \bar U_1, X_{L+1},Z) \hspace{0.05in} \textrm{for all} \hspace{0.05in} S \hspace{0.05in}\textrm{containing 1},
\end{align*}
\begin{align*}
I(\mb{X}_S;\mb{U}_S | \mb{U}_{S^c \setminus \{1\}}, U_1 , X_{L+1},Z) = I(\mb{X}_S;\mb{U}_S | \mb{U}_{S^c \setminus \{1\}}, \bar U_1 , X_{L+1},Z) \hspace{0.05in} \textrm{for all nonempty} \hspace{0.05in} S \hspace{0.05in}\textrm{not containing 1},
\end{align*}
and
\[
H(Y | \mb{U}_{1^c}, U_1,X_{L+1},Z) = H(Y | \mb{U}_{1^c}, \bar U_1,X_{L+1},Z).
\]
Since
\[
U_1 \leftrightarrow X_1 \leftrightarrow (\mb{U}_{1^c},\mb{X}_{1^c},X_{L+1},Y,Z),
\]
if we replace $U_1$ by $\bar U_1$ then the resulting $\lambda_i$ is in $\Lambda_i$ and $\mathcal{R}_i^{CI}(\lambda_i)$ remains unchanged. By repeating this procedure for $U_2, \dots, U_L$, we conclude that there exists $\bar {\lambda}_i = (\bar {\mb{U}}, \bar T)$ in $\bar {\Lambda}_i$ such that $\bar T$ is deterministic and $\mathcal{R}_i^{CI}(\lambda_i) = {\mathcal{R}}_i^{CI}(\bar {\lambda}_i)$.

We now turn to general $T$. Consider $\lambda_i = (\mb{U},T)$ in $\Lambda_i$. Let $(\mb{U},t)$ denote the joint distribution of $(\mb{U},T)$ conditioned on $\{T=t\}$. It follows from the deterministic $T$ case that for each $t$ in $\mathcal{T}$, there exists $\bar {\mb{U}}$ such that $(\bar {\mb{U}}, t)$ is in $\bar {\Lambda}_i$ and $\mathcal{R}_i^{CI}(\mb{U},t) = {\mathcal{R}}_i^{CI}(\bar {\mb{U}}, t)$. Hence, on replacing $\mb{U}$ by $\bar {\mb{U}}$ for each $t$ in $\mathcal{T}$, we obtain $(\bar {\mb{U}}, T)$ in $\Lambda_i$ such that $|\bar {\mathcal{U}}_l| \le |\mathcal{X}_l| + 2^{L} - 1$ for all $l$ in $\mathcal{L}$
and $\mathcal{R}_i^{CI}(\mb{U},T) = {\mathcal{R}}_i^{CI}(\bar {\mb{U}}, T)$. Now ${\mathcal{R}}_i^{CI}(\bar {\mb{U}}, T)$ is the set of vectors $(\mb{R},E)$ such that
\begin{align*}
\sum_{l \in S} R_l &\ge I(\mb{X}_S;\mb{U}_S|\mb{U}_{S^c},X_{L+1},Z,T) \hspace{0.05in}\textrm{for all} \hspace{0.05in} S, \hspace{0.05in} \textrm{and}\\
E &\le I(Y;\mb{U},X_{L+1}|Z,T).
\end{align*}
It again follows from the support lemma that there exists $\bar {T}$ with $\bar {\mathcal{T}} \subseteq \mathcal{T}$ such that
\begin{align*}
|\bar {\mathcal{T}}| &\le 2^L, \\
I(\mb{X}_S;\mb{U}_S|\mb{U}_{S^c},X_{L+1},Z,T) &= I(\mb{X}_S;\mb{U}_S|\mb{U}_{S^c},X_{L+1},Z,\bar T), \hspace{0.05in} \textrm{and}\\
I(Y;\mb{U},X_{L+1}|Z,T) &= I(Y;\mb{U},X_{L+1}|Z,\bar T).
\end{align*}
We therefore have that $\bar {\lambda}_i = (\bar {\mb{U}}, \bar T)$ is in $\bar {\Lambda}_i$ and $\mathcal{R}_i^{CI}(\lambda_i) = \bar {\mathcal{R}}_i^{CI}(\bar {\lambda}_i).$ This proves $\mathcal{R}_i^{CI} \subseteq \bar {\mathcal{R}}_i^{CI}$, and hence $\mathcal{R}_i^{CI} = \bar {\mathcal{R}}_i^{CI}$ because the reverse containment trivially holds.

For part (b), it suffices to show that $\bar {\mathcal{R}}_i^{CI}$ is closed. Consider any sequence $\left(\mb{R}^{(n)},E^{(n)}\right)$ in $\bar {\mathcal{R}}_i^{CI}$ that converges to $(\mb{R},E)$. Since conditional mutual information is a continuous function, $\bar {\Lambda}_i$ is a compact set. Hence, there exists a sequence $\lambda_i^{(n)} = \left(\mb{U}^{(n)},T^{(n)}\right)$ in $\bar {\Lambda}_i$ that converges to $\lambda_i = \left(\mb{U},T\right)$ in $\bar {\Lambda}_i$ such that $\left(\mb{R}^{(n)},E^{(n)}\right)$ is in $\bar {\mathcal{R}}_i^{CI}\left (\lambda_i^{(n)}\right)$, i.e.,
\begin{align*}
\sum_{l \in S} R_l^{(n)} &\ge I\left(\mb{X}_S;\mb{U}^{(n)}_S \bigr|\mb{U}^{(n)}_{S^c},X_{L+1},Z,T^{(n)}\right) \hspace{0.05in}\textrm{for all} \hspace{0.05in} S, \hspace{0.05in} \textrm{and}\\
E^{(n)} &\le I\left(Y;\mb{U}^{(n)},X_{L+1}\bigr|Z,T^{(n)}\right).
\end{align*}
Again, by the continuity of conditional mutual information, this implies that
\begin{align*}
\sum_{l \in S} R_l &\ge I(\mb{X}_S;\mb{U}_S|\mb{U}_{S^c},X_{L+1},Z,T) \hspace{0.05in}\textrm{for all} \hspace{0.05in} S, \hspace{0.05in} \textrm{and}\\
E &\le I(Y;\mb{U},X_{L+1}|Z,T).
\end{align*}
We thus have that $(\mb{R},E)$ is in $\bar {\mathcal{R}}_i^{CI}$.\\ \\
\Large
{\textbf{Appendix B:}\hspace{0.05in} \textbf{Proof of Theorem 1}}\newline
\normalsize \\
We prove the deterministic $T$ case. The general case follows by time sharing.
Consider any $\lambda_i = (\mb{U},T)$ in $\Lambda_i$ with $T$ being
deterministic. Consider $(\mb{R},E)$ such that
\begin{align}
\sum_{l \in S} R_l &\ge I(\mb{X}_S;\mb{U}_S|\mb{U}_{S^c},X_{L+1},Z) \hspace{0.05in} \textrm{for all} \hspace{0.05in} S \subseteq \mathcal{L}, \hspace{0.05in} \textrm{and} \\
E &\le I(Y;\mb{U},X_{L+1}|Z).
\end{align}
It suffices to show that $(\mb{R},E)$ belongs to the rate-exponent region $\mathcal{R}^{CI}$.

Consider a sufficiently large block length $n$, $\epsilon > 0$, and $\mu > 0$. For each $l$ in $\mathcal{L}$, let $\bar R_l = I(X_l;U_l) + \alpha$, where $\alpha > 0$. To construct the codebook of encoder $l$, we first generate $2^{n \bar R_l}$ independent codewords $U_l^n$, each according to $\prod_{i=1}^n P_{U_l}(u_{li})$, and then distribute them uniformly into $2^{n (R_l+\epsilon)}$ bins. The codebooks and the bin assignments are revealed to the encoders and the detector. The encoding is done in two steps: quantization and binning. The encoder $l$ first quantizes $X_l^n$ by selecting a codeword $U_l^n$ that is jointly $\mu$-typical with it. We adopt the typicality notion of Han \cite{Han1}. If there is more than one such codeword, then the encoder $l$ selects one of them arbitrarily. If there is no such codeword, it selects an arbitrary codeword.
The encoder then sends to the detector the index of the bin to which the codeword $U_l^n$ belongs. In order to be consistent with our earlier notation, we denote this encoding function by $f_l^{(n)}$. It is clear that the rate constraints are satisfied, i.e.,
\begin{align}
\frac{1}{n} \log \left|f_l^{(n)}(X_l^n)\right| = R_l +\epsilon \hspace{0.05in} \textrm{for all} \hspace{0.05in}  l  \hspace{0.05in} \textrm{in}  \hspace{0.05in} \mathcal{L}.
\end{align}

The next lemma is a standard achievability result in distributed source coding.
\begin{Lem}
For any $\delta > 0, \epsilon >0, \mu > 0,$ and all sufficiently large $n$, there exists a function
$$
\varphi^{(n)} : \prod_{l=1}^L \left \{1,\dots, 2^{n(R_l+\epsilon)} \right \} \times \mathcal{X}^n_{L+1} \times \mathcal{Z}^n \mapsto \prod_{l=1}^L \mathcal{U}_l^n
$$
such that (a)
if
\[
V \triangleq \left \{\mb{U}^n, X_{L+1}^n,Y^n, Z^n \hspace{0.05in} \textrm{are jointly $\mu$-typical under $H_0$} \right \},
\]
then
$P(V) \ge 1 - \delta$; and (b)
\begin{align*}
p_e \triangleq P\left (\varphi^{(n)} \left ((f_l^{(n)}(X_l^n))_{l \in \mathcal{L}},
       X_{L+1}^n, Z^n \right) \neq \mb{U}^n \right) \le \delta.
\end{align*}
\end{Lem}
One can prove this lemma using standard random coding arguments. See \cite{Berger1,Tung,Gastper} for proofs of similar results. Applying this lemma to the hypothesis testing problem at hand, we have
\begin{align}
\frac{1}{n} I &\left(\left (f_l^{(n)} \left ({X}_l^n \right)\right)_{l \in \mathcal{L}},X_{L+1}^n;{Y}^n \Bigr|Z^n\right) \nonumber\\
&= \frac{1}{n}H\left (Y^n|Z^n \right) - \frac{1}{n}H \left({Y}^n \Bigr|\left (f_l^{(n)} \left ({X}_l^n \right)\right)_{l \in \mathcal{L}},X_{L+1}^n,Z^n\right)\nonumber\\
&= H\left (Y|Z \right) + \frac{1}{n} H\left (\left (f_l^{(n)} \left ({X}_l^n \right)\right)_{l \in \mathcal{L}} \Bigr|X_{L+1}^n,Z^n \right) - \frac{1}{n}H \left(\left (f_l^{(n)} \left ({X}_l^n \right)\right)_{l \in \mathcal{L}},{Y}^n \Bigr |X_{L+1}^n,Z^n\right).
\end{align}
We can lower bound the second term in (39) as
\begin{align}
\frac{1}{n} H \left (\left (f_l^{(n)} \left ({X}_l^n \right)\right)_{l \in \mathcal{L}} \Bigr|X_{L+1}^n,Z^n \right)\nonumber &=\frac{1}{n} I\left (\left (f_l^{(n)} \left ({X}_l^n \right)\right)_{l \in \mathcal{L}};\mb{U}^n \Bigr|X_{L+1}^n,Z^n \right) \nonumber\\
&=\frac{1}{n} H\left (\mb{U}^n|X_{L+1}^n,Z^n \right) - \frac{1}{n} H \left(\mb{U}^n \Bigr| \left (f_l^{(n)} \left ({X}_l^n \right)\right)_{l \in \mathcal{L}}, X_{L+1}^n,Z^n\right ) \nonumber\\
&\ge \frac{1}{n} H\left (\mb{U}^n|X_{L+1}^n,Z^n \right) - \frac{1}{n} H \left(\mb{U}^n \Bigr|\varphi^{(n)} \left ( \left (f_l^{(n)} \left ({X}_l^n \right)\right)_{l \in \mathcal{L}}, X_{L+1}^n,Z^n \right )\right ) \\
&\ge \frac{1}{n} H\left (\mb{U}^n|X_{L+1}^n,Z^n \right) - \frac{1}{n} H_b(p_e) - p_e \sum_{l=1}^L\log |\mathcal{U}_l| \\
&\ge \frac{1}{n} H\left (\mb{U}^n|X_{L+1}^n,Z^n \right) - \frac{1}{n}  - \delta \sum_{l=1}^L\log |\mathcal{U}_l|,
\end{align}
where
\begin{enumerate}
\item[(40)] follows from data processing inequality \cite[Theorem 2.8.1]{Cover},
\item[(41)] follows from Fano's inequality \cite[Theorem 2.10.1]{Cover}, and
\item[(42)] follows Lemma 6(b) and the fact that $H_b(p_e) \le 1$.
\end{enumerate}
The third term in (39) can be upper bounded as
\begin{align}
\frac{1}{n}H \left(\left (f_l^{(n)} \left ({X}_l^n \right)\right)_{l \in \mathcal{L}},{Y}^n \Bigr|X_{L+1}^n,Z^n\right) &\le \frac{1}{n}H \left(\mb{U}^n, \left (f_l^{(n)} \left ({X}_l^n \right)\right)_{l \in \mathcal{L}},{Y}^n \Bigr|X_{L+1}^n,Z^n\right) \nonumber\\
&= \frac{1}{n}H \left(\mb{U}^n,{Y}^n|X_{L+1}^n,Z^n\right).
\end{align}
On applying bounds (42) and (43) into (39), we obtain
\begin{align}
\frac{1}{n} I &\left(\left (f_l^{(n)} \left ({X}_l^n \right)\right)_{l \in \mathcal{L}},X_{L+1}^n;{Y}^n \Bigr|Z^n\right) \nonumber\\
&\ge H\left (Y|Z \right) + \frac{1}{n} H\left (\mb{U}^n|X_{L+1}^n,Z^n \right) - \frac{1}{n}H \left(\mb{U}^n,{Y}^n|X_{L+1}^n,Z^n\right) - \frac{1}{n}  - \delta \sum_{l=1}^L\log |\mathcal{U}_l|\nonumber\\
&= H\left (Y|Z \right) - \frac{1}{n}H \left({Y}^n|\mb{U}^n,X_{L+1}^n,Z^n\right) - \frac{1}{n}  - \delta \sum_{l=1}^L\log |\mathcal{U}_l| \nonumber\\
&= H\left (Y|Z \right) - \frac{1}{n}H \left({Y}^n,1_V|\mb{U}^n,X_{L+1}^n,Z^n\right) - \frac{1}{n}  - \delta \sum_{l=1}^L\log |\mathcal{U}_l| \nonumber\\
&= H\left (Y|Z \right) - \frac{1}{n}H \left(1_V|\mb{U}^n,X_{L+1}^n,Z^n\right) - \frac{1}{n}H \left(Y^n|\mb{U}^n,X_{L+1}^n,Z^n, 1_V\right) - \frac{1}{n}  - \delta \sum_{l=1}^L\log |\mathcal{U}_l| \nonumber\\
&\ge H\left (Y|Z \right) - \frac{1}{n} - \frac{1}{n}H \left(Y^n|\mb{U}^n,X_{L+1}^n,Z^n, 1_V=1\right) P(V) \nonumber\\
&\hspace{0.5in}- \frac{1}{n}H \left(Y^n|\mb{U}^n,X_{L+1}^n,Z^n, 1_V=0\right) P(V^c)  - \frac{1}{n}  - \delta \sum_{l=1}^L\log |\mathcal{U}_l|\\
&\ge H\left (Y|Z \right)  - \frac{1}{n}H \left(Y^n|\mb{U}^n,X_{L+1}^n,Z^n, 1_V=1\right) - \frac{2}{n} -\delta \log |\mathcal{Y}|-\delta \sum_{l=1}^L\log |\mathcal{U}_l|,
\end{align}
where
\begin{enumerate}
\item[(44)] follows from the fact that $H\left(1_V|\mb{U}^n,X_{L+1}^n,Z^n\right) \le 1$, and
\item[(45)] follows from Lemma 6(a) and the facts that
\begin{align*}
\frac{1}{n}H \left(Y^n|\mb{U}^n,X_{L+1}^n,Z^n, 1_V=0\right) &\le \log |\mathcal{Y}|, \\
P(V) &\le 1.
\end{align*}
\end{enumerate}
We now proceed to upper bound the second term in (45). Let ${T}^n_{\mu}(\mb{U}X_{L+1}YZ)$ be the set of all jointly $\mu$-typical $(\mb{u}^n,x_{L+1}^n,y^n,z^n)$ sequences. We need the following lemma.
\begin{Lem}
\cite[Lemma 1(d)]{Han1} If $n$ is sufficiently large, then for any $(\mb{u}^n,x_{L+1}^n,y^n,z^n)$ in ${T}^n_{\mu}(\mb{U}X_{L+1}YZ)$, we have
\begin{align*}
P_{Y^n|\mb{U}^n,X_{L+1}^n,Z^n}(y^n|\mb{u}^n,x_{L+1}^n,z^n) \ge \exp \left[ -n \left (H \left(Y|\mb{U},X_{L+1},Z\right) + 2\mu \right)\right] .
\end{align*}
\end{Lem}
Using this lemma, we obtain
\begin{align}
\frac{1}{n}H \left(Y^n|\mb{U}^n,X_{L+1}^n,Z^n, 1_V=1\right) &= - \frac{1}{n} \sum_{{T}^n_{\mu}(\mb{U}X_{L+1}YZ)}  P_{\mb{U}^n,X_{L+1}^n,Y^n,Z^n| 1_V=1} \log P_{Y^n|\mb{U}^n,X_{L+1}^n,Z^n, 1_V=1} \nonumber\\
&=- \frac{1}{n} \sum_{{T}^n_{\mu}(\mb{U}X_{L+1}YZ)}  P_{\mb{U}^n,X_{L+1}^n,Y^n,Z^n| 1_V=1} \log \frac{P_{Y^n|\mb{U}^n,X_{L+1}^n,Z^n}}{P_{1_V=1|\mb{U}^n,X_{L+1}^n,Z^n}} \nonumber\\
&\le \sum_{{T}^n_{\mu}(\mb{U}X_{L+1}YZ)} P_{\mb{U}^n,X_{L+1}^n,Y^n,Z^n| 1_V=1} \left (H \left(Y|\mb{U},X_{L+1},Z\right) + 2\mu \right) \nonumber\\
&= H \left(Y|\mb{U},X_{L+1},Z\right) + 2\mu.
\end{align}
Substituting (46) into (45) gives
\begin{align}
\frac{1}{n} I \left(\left (f_l^{(n)} \left ({X}_l^n \right) \right)_{l \in \mathcal{L}},X_{L+1}^n;{Y}^n \Bigr|Z^n\right)
&\ge I \left(Y;\mb{U},X_{L+1}|Z\right) - \frac{2}{n}  - 2 \mu - \delta \log |\mathcal{Y}|  - \delta \sum_{l=1}^L\log |\mathcal{U}_l| \nonumber\\
&\ge E - 3 \mu - \delta \log |\mathcal{Y}|  - \delta \sum_{l=1}^L\log |\mathcal{U}_l|,
\end{align}
where the last inequality follows from (37) and the fact that $n$ can be made arbitrarily large. We conclude from (38) and (47) that $$\left(R_1+\epsilon,\dots,R_L+\epsilon,E-3 \mu - \delta \log |\mathcal{Y}|  - \delta \sum_{l=1}^L\log |\mathcal{U}_l|\right)$$ is in ${\mathcal{R}}_{*}^{CI}$. Since this is true for any $\delta > 0, \epsilon >0,$ and $\mu > 0,$ we have that $(\mb{R},E)$ is in $\overline{{\mathcal{R}}_{*}^{CI}}$. This together with Corollary 1 implies that $(\mb{R},E)$ is in ${\mathcal{R}}^{CI}.$
\\ \\ \Large
{\textbf{Appendix C:}\hspace{0.05in} \textbf{Proof of Theorem 2}}\newline
\normalsize \\
Suppose $(\mb{R},E)$ is in $\mathcal{R}_{*}^{CI}$. Then there exists a block length $n$ and encoders $f_l^{(n)}$ such that
\begin{align}
R_l &\ge \frac{1}{n}\log \left|f_l^{(n)}\left({X}_l^n\right)\right| \hspace{0.05in} \textrm{for all \emph{l} in} \hspace{0.05in} \mathcal{L}, \hspace {0.05 in} \textrm{and} \\
E &\le  \frac{1}{n}I\biggr(\left (f_l^{(n)} \left ({X}_l^n \right)\right)_{l \in \mathcal{L}}, X_{L+1}^n;{Y}^n \Bigr|Z^n\biggr).
\end{align}
Consider any $X$ in $\chi$. Let $T$ be a time sharing random variable uniformly distributed over $\{1,\dots,n\}$ and independent of $(\mb{X}^n, X_{L+1}^n,X^n, Y^n, Z^n)$. Define
\begin{align*}
X_l &= X_l^n(T) \hspace {0.05 in} \textrm{for each} \hspace {0.05 in} l  \hspace {0.05 in}\textrm{in} \hspace {0.05 in} \mathcal{L} \cup \{L+1\},\\
X &= X^n(T), \\
Y &= Y^n(T), \\
Z &= Z^n(T), \\
U_l & = \left(f_l^{(n)}\left({X}_l^n\right), X^n(1 : T-1), X^n_{L+1}(T^c), Z^n(T^c) \right)  \hspace {0.05 in} \textrm{for each} \hspace {0.05 in} l  \hspace {0.05 in}\textrm{in} \hspace {0.05 in} \mathcal{L}, \hspace {0.05 in}\textrm{and} \\
W &= \left (X^n(T^c),X^n_{L+1}(T^c),Z^n(T^c) \right ).
\end{align*}
It is easy to verify that $\lambda_o = (\mb{U}, W, T)$ is in $\Lambda_o$ and
\[
X \leftrightarrow (\mb{X},X_{l+1}, Y,Z) \leftrightarrow \lambda_o.
\]
It suffices to show that $(\mb{R}, E)$ is in $\mathcal{R}^{CL}_o(X,\lambda_o)$. We obtain the following from (49)
\begin{align}
E &\le  \frac{1}{n} I\left( \left (f_l^{(n)} \left ({X}_l^n \right)\right)_{l \in \mathcal{L}}, X_{L+1}^n;{Y}^n \Bigr|Z^n\right) \nonumber\\
&= \frac{1}{n} \left [H(Y^n|Z^n) - H\left (Y^n\Bigr |\left (f_l^{(n)} \left ({X}_l^n \right)\right)_{l \in \mathcal{L}},  X_{L+1}^n,Z^n \right) \right ]\nonumber\\
&=\frac{1}{n}\sum_{i=1}^n \left [ H(Y^n(i)|Z^n(i)) - H\left (Y^n(i)\Bigr |\left (f_l^{(n)} \left ({X}_l^n \right)\right)_{l \in \mathcal{L}}, Y^n(1:i-1),  X_{L+1}^n,Z^n \right)\right ] \nonumber\\
&\le \frac{1}{n}\sum_{i=1}^n \left [ H(Y^n(i)|Z^n(i)) - H\left (Y^n(i)\Bigr |\left (f_l^{(n)} \left ({X}_l^n \right)\right)_{l \in \mathcal{L}}, Y^n(1:i-1), X^n(1:i-1), X_{L+1}^n,Z^n \right)\right ] \\
&= \frac{1}{n}\sum_{i=1}^n \left [ H(Y^n(i)|Z^n(i)) - H\left (Y^n(i)\Bigr |\left (f_l^{(n)} \left ({X}_l^n \right)\right)_{l \in \mathcal{L}}, X^n(1:i-1), X_{L+1}^n,Z^n \right)\right ] \\
&= \frac{1}{n}\sum_{i=1}^n  I \left (Y^n(i) ; \left (f_l^{(n)} \left ({X}_l^n \right)\right)_{l \in \mathcal{L}}, X^n(1:i-1),X^n_{L+1}(i^c),Z^n(i^c), X_{L+1}^n(i) \Bigr| Z^n(i) \right) \nonumber\\
&=  I \left (Y^n(T) ; \mb{U},X^n_{L+1}(T) | Z^n(T), T \right) \nonumber\\
&=  I \left (Y ; \mb{U} , X_{L+1}| Z, T \right),\nonumber
\end{align}
where
\begin{enumerate}
\item[(50)] follows from conditioning reduces entropy, and
\item[(51)] follows because of the Markov chain
\[
Y^n(1:i-1) \leftrightarrow \left (X^n(1:i-1),Z^n(1:i-1)\right) \leftrightarrow \left (\left (f_l^{(n)} \left ({X}_l^n \right)\right)_{l \in \mathcal{L}},X_{L+1}^n, Y^n(i), Z^n(i:n) \right).
\]
\end{enumerate}
Now let $S \subseteq \mathcal{L}$. Then (48) implies
\begin{align}
n \sum_{l \in S} R_l &\ge \sum_{l \in S} \log \left|f_l^{(n)}\left({X}_l^n\right)\right| \nonumber\\
&\ge \sum_{l \in S} H \left (f_l^{(n)}\left({X}_l^n\right)\right) \nonumber\\
&\ge H \left (\left (f_l^{(n)}\left({X}_l^n\right) \right )_{l \in S}\right) \nonumber\\
&\ge H \left (\left (f_l^{(n)}\left({X}_l^n\right) \right )_{l \in S} \Bigr | \left (f_l^{(n)}\left({X}_l^n\right) \right )_{l \in S^c}, X_{L+1}^n,Z^n\right) \\
&= I \left (X^n, \mb{X}^n_S ; \left (f_l^{(n)}\left({X}_l^n\right) \right )_{l \in S} \Bigr | \left (f_l^{(n)}\left({X}_l^n\right) \right )_{l \in S^c}, X_{L+1}^n,Z^n\right) \nonumber\\
&= I \left (X^n ; \left (f_l^{(n)}\left({X}_l^n\right) \right )_{l \in S} \Bigr | \left (f_l^{(n)}\left({X}_l^n\right) \right )_{l \in S^c}, X_{L+1}^n,Z^n\right) \nonumber\\
&\hspace{0.5in}+ I \left (\mb{X}^n_S ; \left (f_l^{(n)}\left({X}_l^n\right) \right )_{l \in S} \Bigr | \left (f_l^{(n)}\left({X}_l^n\right) \right )_{l \in S^c}, X^n, X_{L+1}^n,Z^n\right) \nonumber\\
&= \sum_{i=1}^n I \left (X^n(i) ; \left (f_l^{(n)}\left({X}_l^n\right) \right )_{l \in S} \Bigr | \left (f_l^{(n)}\left({X}_l^n\right) \right )_{l \in S^c}, X^n(1:i-1), X_{L+1}^n,Z^n\right) \nonumber\\
&\hspace{0.5in} + \sum_{l \in S}I \left ({X}^n_l ; f_l^{(n)}\left({X}_l^n\right) \Bigr | X^n, X_{L+1}^n,Z^n\right),
\end{align}
where
\begin{enumerate}
\item[(52)] follows from conditioning reduces entropy, and
\item[(53)] follows because $X$ is in $\chi$.
\end{enumerate}
We next lower bound the second sum in (53).
\begin{align}
I &\left ({X}^n_l ; f_l^{(n)}\left({X}_l^n\right) \Bigr | X^n, X_{L+1}^n,Z^n\right) \nonumber\\
&= \sum_{i=1}^n I \left ({X}^n_l(i) ; f_l^{(n)}\left({X}_l^n\right) \Bigr | X^n, {X}^n_l(1:i-1),X_{L+1}^n,Z^n\right)\nonumber \\
&=\sum_{i=1}^n \left [H\left ({X}^n_l(i) \Bigr | X^n, {X}^n_l(1:i-1),X_{L+1}^n,Z^n\right) - H\left ({X}^n_l(i) \Bigr | f_l^{(n)}\left({X}_l^n\right) , X^n, {X}^n_l(1:i-1),X_{L+1}^n,Z^n\right) \right ]\nonumber \\
&\ge \sum_{i=1}^n \left [H\left ({X}^n_l(i) \Bigr | X^n, X_{L+1}^n,Z^n\right) - H\left ({X}^n_l(i) \Bigr | f_l^{(n)}\left({X}_l^n\right) , X^n, X_{L+1}^n,Z^n\right) \right ] \\
&= \sum_{i=1}^n I\left ({X}^n_l(i) ; f_l^{(n)}\left({X}_l^n\right) \Bigr | X^n, X_{L+1}^n,Z^n\right),
\end{align}
where (54) again follows from conditioning reduces entropy. On applying (55) in (53), we obtain
\begin{align}
\sum_{l \in S} R_l &\ge\frac{1}{n}\sum_{i=1}^n \biggr [ I \left (X^n(i) ; \left (f_l^{(n)}\left({X}_l^n\right) \right )_{l \in S} \Bigr | \left (f_l^{(n)}\left({X}_l^n\right) \right )_{l \in S^c}, X^n(1:i-1), X_{L+1}^n,Z^n\right) \nonumber\\
 &\hspace{0.4in} + \sum_{l \in S}I \left ({X}^n_l(i) ; f_l^{(n)}\left({X}_l^n\right) \Bigr | X^n, X_{L+1}^n,Z^n\right) \biggr ].
\end{align}
If $S^c$ is nonempty, then continuing from (56) gives
\begin{align*}
\sum_{l \in S} R_l &\ge I \left (X^n(T) ; \mb{U}_S \bigr | \mb{U}_{S^c}, X_{L+1}^n(T),Z^n(T),T\right) \\
&\hspace{0.3in}+ \sum_{l \in S}I \left ({X}^n_l(T) ; U_l \bigr | X^n(T), X_{L+1}^n(T),Z^n(T), X^n(T^c),X_{L+1}^n(T^c), Z^n(T^c),T\right) \\
 &= I \left (X ; \mb{U}_S \bigr | \mb{U}_{S^c}, X_{L+1},Z,T\right) + \sum_{l \in S}I \left ({X}_l ; U_l \bigr | X, W, X_{L+1},Z,T\right).
\end{align*}
Finally if $S = \mathcal{L}$, then
\begin{align}
I &\left (X^n(i) ; \left (f_l^{(n)}\left({X}_l^n\right) \right )_{l \in S} \Bigr | \left (f_l^{(n)}\left({X}_l^n\right) \right )_{l \in S^c}, X^n(1:i-1), X_{L+1}^n,Z^n\right) \nonumber \\
&= I \left (X^n(i) ; \left (f_l^{(n)}\left({X}_l^n\right) \right )_{l \in S} \Bigr |  X^n(1:i-1), X_{L+1}^n,Z^n\right) \nonumber\\
&= I \left (X^n(i) ; \left (f_l^{(n)}\left({X}_l^n\right) \right )_{l \in S}, X^n(1:i-1), X_{L+1}^n(i^c),Z^n(i^c) \Bigr | X_{L+1}^n(i), Z^n(i)\right).
\end{align}
Substituting (57) into (56) yields
\begin{align*}
\sum_{l \in \mathcal{L}} R_l &\ge I \left (X ; \mb{U} \bigr | X_{L+1},Z,T\right) + \sum_{l \in \mathcal{L}}I \left ({X}_l ; U_l \bigr | X, W,X_{L+1},Z,T\right).
\end{align*}
This completes the proof of Theorem 2.\\ \\ \Large
{\textbf{Appendix D:}\hspace{0.05in} \textbf{Proof of Lemma 2}}\newline
\normalsize \\
It suffices to show that (C6) implies (C7). The other direction immediately follows by letting $\epsilon \rightarrow 0.$ We can assume without loss of generality that $|\mathcal{X}| \ge 2$ because the lemma trivially holds otherwise. Let $\mathcal{X} = \{1,2,\dots,|\mathcal{X}|\}$ be the alphabet set of $X$. Let $P_i$ be the $i$th row of the stochastic matrix $P_{Y|X}$ corresponding to $X = i$. We need the following lemma.
\begin{Lem}
If (C6) holds, then rows $P_i$ corresponding to positive $P_X(i)$ are distinct.
\end{Lem}
\begin{proof}
The proof is by contradiction. Suppose that $P_X(1)$ and $P_X(2)$ are positive and $P_1 = P_2.$ Let us define a random variable $U$ as
\[
U \triangleq \left\{
\begin{array}{l l}
  2 & \quad \mbox{if $X = 1,2$}\\
  X & \quad \mbox{otherwise.}\\ \end{array} \right.
\]
The stochastic matrix $P_{X|U}$ has
\begin{align*}
P_{X|U}(1|2) &= \frac{P_X(1)}{P_X(1)+P_X(2)}, \\
P_{X|U}(2|2) &= \frac{P_X(2)}{P_X(1)+P_X(2)}, \hspace{0.05in} \textrm{and} \\
P_{X|U}(i|i) &= 1 \hspace{0.05in} \textrm{ for all} \hspace{0.05in} i \hspace{0.05in} \textrm{in} \hspace{0.05in} \left \{3,4,\dots,|\mathcal{X}| \right\}.
\end{align*}
It is easy to see that $Y, X,$ and $U$ form a Markov chain
\begin{align}
Y \leftrightarrow X \leftrightarrow U.
\end{align}
We now have
\begin{align}
H(Y|U) &= \sum_{i = 2}^{|\mathcal{X}|} H(Y|U=i) P_U(i) \nonumber\\
&= H(Y|U=2) P_U(2)+\sum_{i = 3}^{|\mathcal{X}|} H(Y|U=i) P_U(i) \nonumber\\
&= H \left (\sum_{j=1}^{|\mathcal{X}|} P_j P_{X|U}(j|2)\right) P_U(2)+\sum_{i = 3}^{|\mathcal{X}|} H \left (\sum_{j=1}^{|\mathcal{X}|} P_j P_{X|U}(j|i)\right) P_U(i) \nonumber\\
&= H \left (P_2\right) P_U(2)+\sum_{i = 3}^{|\mathcal{X}|} H \left (P_i\right) P_U(i) \nonumber\\
&= \sum_{i = 2}^{|\mathcal{X}|} H \left (P_i \right) P_U(i),
\end{align}
and
\begin{align}
H(Y|X) &= \sum_{j = 1}^{|\mathcal{X}|} H(P_j) P_X(j)   \nonumber\\
&= \sum_{j = 1}^{|\mathcal{X}|} H(P_j) \left(\sum_{i=2}^{|\mathcal{X}|}P_{X|U}(j|i)P_U(i)\right)   \nonumber\\
&= \sum_{i = 2}^{|\mathcal{X}|} P_U(i) \sum_{j = 1}^{|\mathcal{X}|} P_{X|U}(j|i) H \left (P_j \right) \nonumber\\
&= P_U(2) \sum_{j = 1}^{|\mathcal{X}|} P_{X|U}(j|2) H \left (P_j \right) + \sum_{i = 3}^{|\mathcal{X}|} P_U(i) \sum_{j = 1}^{|\mathcal{X}|} P_{X|U}(j|i) H \left (P_j \right) \nonumber\\
&= P_U(2) H \left (P_2 \right) + \sum_{i = 3}^{|\mathcal{X}|} P_U(i) H \left (P_i \right) \nonumber\\
&= \sum_{i = 2}^{|\mathcal{X}|}P_U(i) H \left (P_i \right).
\end{align}
Now (58) through (60) together imply that $I(X;Y|U) = 0$, and hence $Y \leftrightarrow U \leftrightarrow X.$ However,
\begin{align*}
H(X|U) &= \sum_{i = 2}^{|\mathcal{X}|} H(X|U=i) P_U(i) \nonumber\\
&=  H(X|U=2) P_U(2) \nonumber\\
&= H_b\left(\frac{P_X(1)}{P_X(1)+P_X(2)}\right) (P_X(1)+P_X(2)) \nonumber\\
&> 0,
\end{align*}
which contradicts our assumption that (C6) holds.
\end{proof}

Consider any $U$ that satisfies the Markov chain
\[
U \leftrightarrow X \leftrightarrow Y.
\]
We can assume without loss of generality that $P_U(u)$ is positive for all $u$ in $\mathcal{U}$ because only positive $P_U(u)$ contributes to $H(X|U)$ and $I(X;Y|U)$ in conditions (C6) and (C7). Then
\begin{align}
I(X;Y|U) &= H(Y|U) - H(Y|X) \nonumber\\
&= \sum_{u \in \mathcal{U}} H(Y|U=u) P_U(u) - \sum_{i = 1}^{|\mathcal{X}|} P_X(i) H(P_i) \nonumber\\
&= \sum_{u \in \mathcal{U}} H \left (\sum_{i=1}^{|\mathcal{X}|} P_i P_{X|U}(i|u) \right) P_U(u) - \sum_{i = 1}^{|\mathcal{X}|} \left(\sum_{u \in \mathcal{U}} P_{X|U}(i|u)P_U(u)\right) H(P_i) \nonumber\\
&= \sum_{u \in \mathcal{U}} P_U(u) \left [H \left (\sum_{i=1}^{|\mathcal{X}|} P_i P_{X|U}(i|u) \right)- \sum_{i = 1}^{|\mathcal{X}|} P_{X|U}(i|u)H(P_i) \right ] \nonumber\\
&= \sum_{u \in \mathcal{U}} P_U(u) T\left(P_{X|U}(.|u)\right),
\end{align}
where (61) follows by setting
\[
T\left(P_{X|U}(.|u)\right) \triangleq H \left (\sum_{i=1}^{|\mathcal{X}|} P_i P_{X|U}(i|u) \right)- \sum_{i = 1}^{|\mathcal{X}|} P_{X|U}(i|u)H(P_i).
\]
Since entropy is a strictly concave and continuous function, $T$ is a nonnegative continuous function of $P_{X|U}(.|u)$. Moreover, for any $u$ in $\mathcal{U}$, $P_{X|U}(i|u)=0$ for all $i$ in $\mathcal{X}$ such that $P_X(i)=0$. Let $\mathcal{P}$ denote the set of all such $P_{X|U}(.|u)$. Define
\[
\gamma(\delta) \triangleq \sup_{P \in \mathcal{P}} \{H(P) : T(P) \le \delta \}.
\]
It now follows from Lemma 8 that if $T(P)=0$ for some $P$ in $\mathcal{P}$, then $P$ must be a point mass and hence $H(P)=0$. Therefore, $\gamma(0)=0$. We next show that $\gamma$ is continuous at $0$. Consider a nonnegative sequence $\delta_n \rightarrow 0$. Then there exists a sequence of distributions $P_n$ in $\mathcal{P}$ such that
\begin{align}
T(P_n) &\le \delta_n \\
H(P_n) &\ge \frac{\gamma(\delta_n)}{2}.
\end{align}
Now, since the set of all distributions on $\mathcal{X}$ is a compact set, by considering a subsequence, we can assume without loss of generality that $P_n$ converges to $P$ in $\mathcal{P}$. By letting $n \rightarrow \infty$ in (62), we obtain that $T(P) = 0$, i.e., $P$ is a point mass. Therefore, $H(P)=0$. It now follows from (63) that $\gamma(\delta_n)\rightarrow 0=\gamma(0)$ as $n \rightarrow \infty$. Hence, $\gamma$ is continuous at $0$.

Fix $0 < \epsilon < \log |\mathcal{X}|$ (condition (C7) is always true for $\epsilon \ge \log |\mathcal{X}|$). Choose $\epsilon_1 > 0$ such that $\gamma\left (\epsilon_1/{\log |\mathcal{X}|}\right)+\epsilon_1 = \epsilon$. Set $\delta = \left(\epsilon_1/{\log |\mathcal{X}|}\right)^2$. Let $I(X;Y|U) \le \delta$. Define the sets
\begin{align*}
\mathcal{U}_1 &\triangleq\left \{u \in \mathcal{U} : T(u) \le \sqrt{\delta}\right\} \hspace{0.05in} \textrm{and} \hspace{0.05in} \mathcal{U}_2 \triangleq \mathcal{U} \setminus \mathcal{U}_1.
\end{align*}
Note that $\mathcal{U}_1$ is nonempty because $\delta < 1$.  We now have
\begin{align*}
\delta &\ge I(X;Y|U) \\
&= \sum_{\mathcal{U}} P_U(u) T(u) \\
&\ge \sum_{\mathcal{U}_2} P_U(u) T(u) \\
&> \sqrt{\delta} \sum_{\mathcal{U}_2} P_U(u),
\end{align*}
which implies $$\sum_{\mathcal{U}_2} P_U(u) < \sqrt{\delta}.$$ Hence,
\begin{align*}
H(X|U) &= \sum_{\mathcal{U}_1} H(X|U=u) P_U(u) + \sum_{\mathcal{U}_2} H(X|U=u) P_U(u) \nonumber\\
&< \gamma \bigr(\sqrt{\delta} \bigr)  +  \sqrt{\delta} \log |\mathcal{X}|\\
&= \gamma\left (\epsilon_1/{\log |\mathcal{X}|}\right)+\epsilon_1\\
&= \epsilon.
\end{align*}


\begin{thebibliography}{10}
\providecommand{\url}[1]{#1} \csname url@rmstyle\endcsname
\providecommand{\newblock}{\relax} \providecommand{\bibinfo}[2]{#2}
\providecommand\BIBentrySTDinterwordspacing{\spaceskip=0pt\relax}
\providecommand\BIBentryALTinterwordstretchfactor{4}
\providecommand\BIBentryALTinterwordspacing{\spaceskip=\fontdimen2\font
plus \BIBentryALTinterwordstretchfactor\fontdimen3\font minus
  \fontdimen4\font\relax}
\providecommand\BIBforeignlanguage[2]{{%
\expandafter\ifx\csname l@#1\endcsname\relax
\typeout{** WARNING: IEEEtran.bst: No hyphenation pattern has been}%
\typeout{** loaded for the language `#1'. Using the pattern for}%
\typeout{** the default language instead.}%
\else \language=\csname l@#1\endcsname \fi #2}}


\bibitem{He}
T. He and L. Tong, ``Detection of information flows," \emph{IEEE Trans. Inf. Theory}, vol. 54, no. 11, pp. 4925-4945, Nov. 2008.
\bibitem{Ameya}
A. Agaskar, T. He, and L. Tong, ``Distributed detection of multi-hop information flows with fusion capacity constraints," \emph{IEEE Trans. Sig. Proc.},  vol. 58, no. 6, pp. 3373-3383, June 2010.
\bibitem{Berger}
T. Berger, ``Decentralized estimation and decision theory," in \emph{IEEE 7th Spring Workshop on Inf. Theory}, Mt. Kisco, NY, Sept. 1979.
\bibitem{Han}
T. S. Han  and S. Amari, ``Statistical inference under multiterminal data compression," \emph{IEEE Trans. Inf. Theory}, vol. 44, no. 6, pp. 2300-2324, Oct. 1998.
\bibitem{Ahl}
R. Ahlswede and I. Csisz\'{a}r, ``Hypothesis testing with communication constraints," \emph{IEEE Trans. Inf. Theory}, vol. 32, no. 4, pp. 533-542, July 1986.
\bibitem{Han1}
T. S. Han, ``Hypothesis testing with multiterminal data compression," \emph{IEEE Trans. Inf. Theory}, vol. 33, no. 6, pp. 759-772, Nov. 1987.
\bibitem{Han2}
H.~Shimokawa, T.~S.~Han, and S. Amari, ``Error bound of hypothesis testing with data compression," in \emph{IEEE Int. Symp. Inf. Theor. Proc.}, 1994, p. 29.
\bibitem{Slepian}
D. Slepian and J. K. Wolf, ``Noiseless coding of correlated information sources," \emph{IEEE Trans. Inf. Theory}, vol. 19, no. 4, pp. 471-480, July 1973.
\bibitem{Ben}
B. G. Kelly and A. B. Wagner, ``Reliability in source coding with side information," preprint.
\bibitem{Ben2}
B. G. Kelly, A. B. Wagner, and A. Vamvatsikos, ``Error exponents and test channel optimization for the {G}aussian {W}yner-{Z}iv problem," in \emph{IEEE Int. Symp. Inf. Theor. Proc.}, 2008, pp.~414-418.
\bibitem{Ben3}
B. G. Kelly and A. B. Wagner, ``Error exponents and test channel optimization for the {W}yner-{Z}iv problem," in \emph{Proc. 45th Annual Allerton Conference}, 2007.
\bibitem{Kochman}
Y.~Kochman and G.~W.~Wornell, ``On the excess distortion exponent of the quadratic-{G}aussian {W}yner-{Z}iv problem," in \emph{IEEE Int. Symp. Inf. Theor. Proc.}, 2010, p. 36-40.
\bibitem{Csiszar}
I. Csisz\'{a}r, ``Linear codes for sources and source networks: error exponents, universal coding," \emph{IEEE Trans. Inf. Theory}, vol. 28, no. 4, pp. 585-592, July 1982.
\bibitem{Ben1}
B. G. Kelly and A. B. Wagner, ``Improved source coding exponents via Witsenhausen's rate," preprint.
\bibitem{Csiszar1}
I. Csisz\'{a}r and J. K\"{o}rner, ``Graph decomposition: a new key to coding theorems," \emph{IEEE Trans. Inf. Theory}, vol. 27, no. 1, pp. 5-12,
Jan. 1981.
\bibitem{Wagner3}
A. B. Wagner and V. Anantharam, ``An improved outer bound for multiterminal source coding," \emph{IEEE Trans. Inf. Theory}, vol. 54, no. 5, pp. 1919-1937, May 2008.
\bibitem{Wagner1}
S. Tavildar, P. Viswanath, and A. B. Wagner, ``The Gaussian many-help-one distributed source coding problem," \emph{IEEE Trans. Inf. Theory}, vol. 56, no. 1, pp. 564-581, Jan. 2010.
\bibitem{Wagner}
A. B. Wagner, S. Tavildar, and P. Viswanath, ``Rate region of the quadratic Gaussian two-encoder source-coding problem," \emph{IEEE Trans. Inf. Theory}, vol. 54, no. 5, pp. 1938-1961, May 2008.
\bibitem{Wagner2}
A. B. Wagner, ``On distributed compression of linear functions," in \emph{Proc. 46th Annual Allerton Conference}, 2008, pp. 1546-1553.
\bibitem{Ozarow}
Ozarow, ``On a source-coding problem with two channels and three receivers," \emph{Bell Sys. Tech. J.}, vol. 59, no. 10, pp. 1909-1921, 1980.
\bibitem{Wang}
H. Wang and P. Viswanath, ``Vector Gaussian multiple description with two levels of receivers," \emph{IEEE Trans. Inf. Theory}, vol. 55, no. 1, pp. 401-410, Jan. 2009.
\bibitem{Gelfand}
S. I. Gel`fand and M. S. Pinsker, ``Coding of sources on the basis of observations with incomplete information," (in Russian), \emph{Problemy Peredachi Informatsii}, vol. 15, no. 2, pp. 45-57, April-June 1979.
\bibitem{Oohama2005}
Y. Oohama, ``Rate-distortion theory for Gaussian multiterminal source coding systems with several side informations at the decoder," \emph{IEEE Trans. Inf. Theory}, vol. 51, no. 7, pp. 2577-2593, July 2005.
\bibitem{Vinod}
V. Prabhakaran, D. Tse, and K. Ramchandran, ``Rate region of the quadratic Gaussian CEO problem," in \emph{IEEE Int. Symp. Inf. Theor. Proc.}, 2004, p. 117.
\bibitem{Tian}
C. Tian and J. Chen, ``Successive refinement for hypothesis testing and lossless one-helper problem," \emph{IEEE Trans. Inf. Theory}, vol. 54, no. 10, pp. 4666-4681, Oct. 2008.
\bibitem{Berger1}
T. Berger, ``Multiterminal source coding," in \emph{The Information Theory Approach to Communications}, ser. CISM Courses and Lectures, G. Longo, Ed. Springer-Verlag, 1978, vol. 229, pp. 171-231.
\bibitem{Tung}
S. Y. Tung, ``Multiterminal source coding," Ph.D. dissertation, School of Electrical Engineering, Cornell University, Ithaca, NY, May 1978.
\bibitem{Gastper}
M. Gastpar, ``The Wyner-Ziv problem with multiple sources," \emph{IEEE Trans. Inf. Theory}, vol. 50, no. 11, pp. 2762-2768, Nov. 2004.
\bibitem{Chen}
J. Chen, X. Zhang, T. Berger, and S. B. Wicker, ``An upper bound on the sum-rate distortion function and its corresponding rate allocation schemes
for the CEO problem," \emph{IEEE J. Select. Areas Commun.}, vol. 22, no. 6, pp. 977-987, Aug. 2004.
\bibitem{Viswanath}
P. Viswanath, ``Sum rate of a class of Gaussian multiterminal source coding problems," in \emph{Advances in Network Information Theory}, ser. DIMACS
in Discrete Mathematics and Theoretical Computer Science, P. Gupta, G. Kramer, and A. J. van Wijngaarden, Eds. AMS, 2004, vol. 66, pp. 43-60.
\bibitem{Wagner4}
A. B. Wagner, B. G. Kelly, and Y. Altu\u{g}, ``The lossy one-helper conjecture is false," in \emph{Proc. 47th Annual Allerton Conference}, 2009, pp. 716-723.
\bibitem{csiszar}
I. Csisz\'{a}r and J. K\"{o}rner, \emph{Information Theory: Coding Theorems for Discrete Memoryless Systems,} 1st ed. Academic, New York, 1981.
\bibitem{wang}
J. Wang, J. Chen, and X. Wu, ``On the minimum sum rate of Gaussian multiterminal source coding: new proofs," in \emph{IEEE Int. Symp. Inf. Theor. Proc.}, 2009, pp. 1463-1467.
\bibitem{Cover}
T. M. Cover and J. A. Thomas, \emph{Elements of Information Theory}, 2nd ed. New York: John Wiley \& Sons, 2005.
\bibitem{Oohama}
Y. Oohama, ``Gaussian multiterminal source coding," \emph{IEEE Trans. Inf. Theory}, vol. 43, no. 6, pp. 1912-1923, Nov. 1997.
\bibitem{Chao}
C. Tian and J. Chen, ``Remote vector Gaussian source coding with decoder side information under mutual information and distortion constraints," \emph{IEEE Trans. Inf. Theory}, vol. 55, no. 10, pp. 4676-4680, Oct. 2009.
\bibitem{Globerson}
A. Globerson and N. Tishby, ``On the optimality of the Gaussian information bottleneck curve," in \emph{Hebrew Univ. Tech. Report}, 2004.
\bibitem{Rahman}
Md. S. Rahman and A. B. Wagner, ``Vector Gaussian hypothesis testing and lossy one-helper problem," in \emph{IEEE Int. Symp. Inf. Theor. Proc.}, 2009, pp. 968-972.
\end{thebibliography}
\end{document}